\documentclass[a4paper,10pt]{article}

\usepackage[utf8]{inputenc}     

\usepackage{amsfonts}
\usepackage{amsthm}
\usepackage{amssymb}
\usepackage{amsmath}
\usepackage{bbm}
\usepackage{tikz}
\usetikzlibrary{arrows,automata}

\newtheorem{theorem}{Theorem}

\newtheorem{lemma}[theorem]{Lemma}
\newtheorem{proposition}[theorem]{Proposition}
\newtheorem{corollary}[theorem]{Corollary}

\newcommand{\tuple}[1]{\langle #1 \rangle}
\newcommand{\limrun}[0]{\infty}
\newcommand{\emptyword}[0]{\lambda}
\newcommand{\occ}[2]{|#1|_{#2}}
\newcommand{\prefsel}[0]{\upharpoonright}
\newcommand{\noprefsel}[0]{\upharpoonleft}
\newcommand{\prob}[0]{\operatorname{Prob}}
\newcommand{\trans}[1]{\mathchoice{\xrightarrow{#1}}{\xrightarrow{\smash{\lower1pt\hbox{$\scriptstyle #1$}}}}{\xrightarrow{#1}}{\xrightarrow{#1}}}
\newcommand{\weight}[0]{\operatorname{weight}}
\newcommand{\freq}[0]{\operatorname{freq}}

\DeclareUnicodeCharacter{D7}{\times}            
\DeclareUnicodeCharacter{393}{\Gamma}           
\DeclareUnicodeCharacter{394}{\Delta}           
\DeclareUnicodeCharacter{3A6}{\Phi}             
\DeclareUnicodeCharacter{3B1}{\alpha}           
\DeclareUnicodeCharacter{3B2}{\beta}            
\DeclareUnicodeCharacter{3B3}{\gamma}           
\DeclareUnicodeCharacter{3B4}{\delta}           
\DeclareUnicodeCharacter{3B5}{\varepsilon}      
\DeclareUnicodeCharacter{3B7}{\eta}             
\DeclareUnicodeCharacter{3B9}{\iota}            
\DeclareUnicodeCharacter{3BA}{\kappa}           
\DeclareUnicodeCharacter{3BB}{\lambda}          
\DeclareUnicodeCharacter{3BC}{\mu}              
\DeclareUnicodeCharacter{3BD}{\nu}              
\DeclareUnicodeCharacter{3BE}{\xi}              
\DeclareUnicodeCharacter{3C0}{\pi}              
\DeclareUnicodeCharacter{3C1}{\rho}             
\DeclareUnicodeCharacter{3C4}{\tau}             
\DeclareUnicodeCharacter{3C9}{\omega}           
\DeclareUnicodeCharacter{2026}{\ldots}          
\DeclareUnicodeCharacter{2115}{\mathbb{N}}      
\DeclareUnicodeCharacter{211A}{\mathbb{Q}}      
\DeclareUnicodeCharacter{211D}{\mathbb{R}}      
\DeclareUnicodeCharacter{212C}{\mathcal{B}}     
\DeclareUnicodeCharacter{2133}{\mathcal{M}}     
\DeclareUnicodeCharacter{2192}{\rightarrow}     
\DeclareUnicodeCharacter{2203}{\exists}         
\DeclareUnicodeCharacter{2208}{\in}             
\DeclareUnicodeCharacter{2209}{\notin}          
\DeclareUnicodeCharacter{2211}{\sum}            
\DeclareUnicodeCharacter{221E}{\infty}          
\DeclareUnicodeCharacter{222A}{\cup}            
\DeclareUnicodeCharacter{2260}{\neq}            
\DeclareUnicodeCharacter{2282}{\subset}         
\DeclareUnicodeCharacter{2286}{\subseteq}       
\DeclareUnicodeCharacter{228E}{\uplus}          
\DeclareUnicodeCharacter{2291}{\sqsubseteq}     
\DeclareUnicodeCharacter{22C3}{\bigcup}         
\DeclareUnicodeCharacter{22C5}{\cdot}           
\DeclareUnicodeCharacter{22EF}{\cdots}          
\DeclareUnicodeCharacter{27E8}{\langle}         
\DeclareUnicodeCharacter{27E9}{\rangle}         
\DeclareUnicodeCharacter{2A04}{\biguplus}       
\DeclareUnicodeCharacter{2A7D}{\leqslant}       
\DeclareUnicodeCharacter{2A7E}{\geqslant}       
\DeclareUnicodeCharacter{1D49C}{\mathcal{A}}    
\DeclareUnicodeCharacter{1D4A2}{\mathcal{G}}    
\DeclareUnicodeCharacter{1D4AE}{\mathcal{S}}    
\DeclareUnicodeCharacter{1D4AF}{\mathcal{T}}    
\DeclareUnicodeCharacter{1D7CF}{\mathbf{1}}     
\DeclareUnicodeCharacter{1D7D9}{\mathbbm{1}}    

\begin{document}

\title{Preservation of normality \\ by  unambiguous transducers}
\author{Olivier Carton}
\date{\today}
\maketitle

\begin{abstract}
  We consider finite state non-deterministic but unambiguous transducers
  with infinite inputs and infinite outputs, and we consider the property
  of Borel normality of sequences of symbols.  When these transducers are
  strongly connected, and when the input is a Borel normal sequence, the
  output is a sequence in which every block has a frequency given by a
  weighted automaton over the rationals.  We provide an algorithm that
  decides in cubic time whether a unambiguous transducer preserves
  normality.
\end{abstract}


\section{Introduction}

More than one hundred years ago Émile Borel~\cite{Borel09} gave the
definition of \emph{normality}. A real number is normal to an integer base
if, in its infinite expansion expressed in that base, all blocks of digits
of the same length have the same limiting frequency.  Borel proved that
almost all real numbers are normal to all integer bases. However, very
little is known on how to prove that a given number has the property.

The definition of normality was the first step towards a definition of
randomness.  Normality formalizes the least requirements about a random
sequence.  It is indeed expected that in a random sequence, all blocks of
symbols with the same length occur with the same limiting frequency.
Normality, however, is a much weaker notion than the one of purely random
sequences defined by Martin-Löf \cite{DowneyHirschfeldt11}.

The motivation of this work is the study of transformations preserving
randomness, hence preserving normality.  The paper is focused on very
simple transformations, namely those that can be realized by finite-state
machines.  We consider automata with outputs, also known as sequential
transducers, mapping infinite sequences of symbols to infinite sequences of
symbols.  Input deterministic transducers were considered
in~\cite{CartonOrduna} where it was shown that preservation of normality
can be checked in polynomial time for these transducers.  This paper
extends the results to unambiguous transducers, that is, transducers where
each sequence is the input label of exactly one accepting run.  These
machines are of great importance because they coincide with functional
transducers in the following sense.  Each unambiguous transducer is indeed
functional as there is at most one output for each input but is was shown
conversely that each functional transducer is equivalent to some
unambiguous one \cite{ChoffrutGrigorieff99}.

An auxiliary result involving weighted automata is introduced to obtain the
main result. It states that if an unambiguous and strongly connected
transducer is fed with a normal sequence then the frequency of each block
in the output is given by a weighted automaton on rational numbers.  It
implies, in particular, that the frequency of each block in the output
sequence does not depend on the input nor the run labeled with it as long
as this input sequence is normal.  As the output of the run can be the used
transitions, the result shows that each finite run has a limiting frequency
in the run.

Our result result is connected to another strong link between normality and
automata.  Agafonov's theorem \cite{Agafonov68} states that if symbols are
selected in a normal sequence using an oblivious finite state machine, the
resulting sequence is still normal.  Oblivious means here that the choice
of selecting a symbol is based on the state of the machine after reading
the prefix of the sequence before the symbol but not including the symbol
it-self.  We show that our results allows us to recover Agafonov's theorem
about preservation of normality by selection. 

The paper is organized as follows.  Notions of normal sequences and
transducers are introduced in Section~\ref{sec:basic}.  Main results are
stated in Section~\ref{sec:results}.  Proofs of the results and algorithms
are given in Section~\ref{sec:proofs}.  The last section is devoted to
preservation of normality by selection.

\section{Basic Definitions} \label{sec:basic}

\subsection{Normality}

Before giving the formal definition of normality, let us introduce some
simple definitions and notation. Let $A$ be a finite set of \emph{symbols}
that we refer to as the \emph{alphabet}. We write $A^ℕ$ for the set of all
sequences on the alphabet~$A$ and $A^*$ for the set of all (finite) words.
Let us denote by~$μ$ the uniform measure on~$A^ℕ$.  The length of a finite
word $w$ is denoted by $|w|$.  The positions of sequences and words are
numbered starting from~$1$. To denote the symbol at position~$i$ of a
sequence (respectively, word) $w$ we write $w[i]$, and to denote the
substring of $w$ from position $i$ to $j$ inclusive we write $w[i{:}j]$.
The empty word is denoted by~$\emptyword$.  The cardinality of a finite
set~$E$ is denoted by~$\#E$.

Given two words $w$ and $v$ in $A^*$, the number $\occ{w}{v}$ of
occurrences of~$v$ in~$w$ is defined by
\begin{displaymath}
  \occ{w}{v} = \# \{ i : w[i{:} i+|v| - 1] = v \} .
\end{displaymath}
For example, $|abbab|_{ab} = 2$.  Given a word $w ∈ A^+$ and a sequence
$x ∈ A^ℕ$, we refer to the \emph{frequency of $w$ in $x$} as
\begin{displaymath}
   \freq(x,w) = \lim_{n→∞} \frac{\occ{x[1{:}n]}{w}}{n}
\end{displaymath}
when this limit is well-defined.

A sequence $x ∈ A^ℕ$ is \emph{normal} on the alphabet $A$ if for every word
$w ∈ A^*$:
\begin{displaymath}
  \freq(x,w) = \frac{1}{(\#A)^{|w|}}
\end{displaymath}

An occurrence of~$v$ is called \emph{aligned} if its starting position~$i$
(as above) is such that $i-1$ is a multiple of the length of~$v$.  An
alternative definition of normality can be given by counting aligned
occurrences, and it is well-known that they are equivalent (see for example
\cite{BecherCarton18}).  We refer the reader to \cite[Chap.4]{Bugeaud12}
for a complete introduction to normality.

The most famous example of a normal word is due to
Champernowne~\cite{Champernowne33}, who showed in 1933 that the infinite
word obtained from concatenating all the natural numbers (in their usual
order):
\begin{displaymath}
  0123456789101112131415161718192021222324252627282930\hdots
\end{displaymath}
is normal on the alphabet $\{0,1,…,9\}$.

\subsection{Automata and transducers}

In this paper we consider automata with outputs, also known as transducers.
We refer the reader to \cite{PerrinPin04} for a complete introduction to
automata accepting sequences.  Such finite-state machines are used to
realize functions mapping words to words and especially sequences to
sequences.  Each transition of these transducers consumes exactly one
symbol of their input and outputs a word which might be empty.  As many
reasoning ignore the outputs of the transitions, we first introduce
automata.

\begin{figure}[htbp]
  \begin{center}
  \begin{tikzpicture}[->,>=stealth',initial text=,semithick,auto,inner sep=1pt]
  \tikzstyle{every state}=[minimum size=0.4]
  \node[state,initial above,accepting] (q1) at (0,1.5) {$1$};
  \node[state] (q2) at (1.5,1.5) {$2$};
  \node[state] (q3) at (0,0) {$3$};
  \node[state] (q4) at (1.5,0) {$4$};
  \path (q1) edge[out=210,in=150,loop] node {$0$} (q1);
  \path (q1) edge [bend left=15] node {$1$} (q2);
  \path (q2) edge [bend left=15] node {$0$} (q1);
  \path (q1) edge node[swap] {$0$} (q3);
  \path (q3) edge [bend left=15] node {$1$} (q4);
  \path (q4) edge [bend left=15] node {$0,1$} (q3);
  \draw (q4) .. controls (1.1,0.75) and (0.4,0.75) .. node[swap,pos=0.35] {$1$} (q1);
  \end{tikzpicture}
  \end{center}
  \caption{An unambiguous automaton}
  \label{fig:unambiguous}
\end{figure}
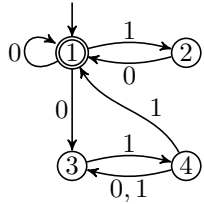

A \emph{(Büchi) automaton} $𝒜$ is a tuple $\tuple{Q,A,Δ,I,F}$ where $Q$ is
the state set, $A$ the alphabet, $Δ ⊆ Q × A × Q$ the transition relation,
$I ⊆ Q$ the set of initial states and $F$ is the set of final states.  A
transition is a tuple $⟨ p,a,q ⟩$ in $Q × A × Q$ and it is written
$p \trans{a} q$.  A \emph{finite run} in~$𝒜$ is a finite sequence of
consecutive transitions,
\begin{displaymath}
   q_0 \trans{a_1} q_1 \trans{a_2} q_2 ⋯ q_{n-1} \trans{a_n} q_n
\end{displaymath}
Its \emph{input} is the word $a_1 a_2 ⋯ a_n$.  An \emph{infinite run}
in~$𝒜$ is a sequence of consecutive transitions,
\begin{displaymath}
  q_0 \trans{a_1} q_1 \trans{a_2} q_2 \trans{a_3} q_3 ⋯ 
\end{displaymath}
A run is \emph{initial} if its first state~$q_0$ is initial, that is,
belongs to~$I$.  A run is called \emph{final} if it visits infinitely often
a final state.  Let us denote by $q \trans{x} \limrun$ the existence of a
final run labeled by~$x$ and starting from state~$q$.  An infinite run is
\emph{accepting} if it is both initial and final.  As usual, an automaton
is \emph{deterministic} if it has only one initial state, that is $\#I = 1$
and if $p \trans{a} q$ and $p \trans{a} q'$ are two of its transitions with
the same starting state and the same label, then $q = q'$.  An automaton is
called \emph{unambiguous} if each sequence is the label of at most one
accepting run.  By definition, deterministic automata are unambiguous but
they are not the only ones as it is shown by Figure~\ref{fig:unambiguous}.

Each automaton~$𝒜$ can be seen as a directed graph~$𝒢$ by ignoring the
labels of its transitions.  We define the \emph{strongly connected
  components} (SCC) of~$𝒜$ as the strongly connected components of~$𝒢$.  An
automaton~$𝒜$ is called \emph{strongly connected} if it has a single
strongly connected component.

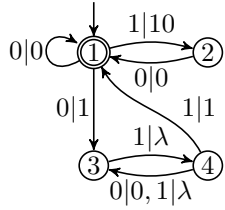
\begin{figure}[htbp]
  \begin{center}
  \begin{tikzpicture}[->,>=stealth',initial text=,semithick,auto,inner sep=1pt]
  \tikzstyle{every state}=[minimum size=0.4]
  \node[state,initial above,accepting] (q1) at (0,1.5) {$1$};
  \node[state] (q2) at (1.5,1.5) {$2$};
  \node[state] (q3) at (0,0) {$3$};
  \node[state] (q4) at (1.5,0) {$4$};
  \path (q1) edge[out=210,in=150,loop] node {$0|0$} (q1);
  \path (q1) edge [bend left=15] node {$1|10$} (q2);
  \path (q2) edge [bend left=15] node {$0|0$} (q1);
  \path (q1) edge node[swap] {$0|1$} (q3);
  \path (q3) edge [bend left=15] node {$1|\emptyword$} (q4);
  \path (q4) edge [bend left=15] node {$0|0,1|\emptyword$} (q3);
  \draw (q4) .. controls (1.1,0.75) and (0.4,0.75) .. node[swap,pos=0.25] {$1|1$} (q1);
  \end{tikzpicture}
  \end{center}
  \caption{An unambiguous transducer}
  \label{fig:transducer}
\end{figure}

A transducer with input alphabet~$A$ and output alphabet~$B$ is informally
an automaton whose labels of transitions are pairs $(a,v)$ in $A × B^*$.
The pair $(a,v)$ is usually written $a|v$ and a transition is thus written
$p \trans{a|v} q$. The symbol~$a$ and the word~$v$ are respectively called
the \emph{input label} and the \emph{output label} of the transition.  More
formally a \emph{transducer}~$𝒯$ is a tuple $⟨ Q,A,B, Δ,I, F⟩$, where $Q$
is a finite set of states, $A$ and $B$ are the input and output alphabets
respectively, $Δ ⊆ Q × A × B^* × Q$ is a finite transition relation and
$I ⊆ Q$ is the set of initial states and $F$ is the set of final states of
the Büchi acceptance condition.  The \emph{input automaton} of a transducer
is the automaton obtained by ignoring the output label of each transition.
The input automaton of the transducer pictured in
Figure~\ref{fig:transducer} is pictured in Figure~\ref{fig:unambiguous}.  A
transducer is called \emph{input deterministic} (respectively,
\emph{unambiguous}) if its input automaton is deterministic (respectively,
unambiguous).

A \emph{finite run} in~$𝒯$ is a finite sequence of consecutive transitions,
\begin{displaymath}
   q_0 \trans{a_1|v_1} q_1
       \trans{a_2|v_2} q_2
       ⋯
       q_{n-1}
       \trans{a_n|v_n} q_n
\end{displaymath}
Its \emph{input} and \emph{output labels} are the words
$a_1 a_2 ⋯ a_n$ and $v_1v_2⋯ v_n$ respectively.

An \emph{infinite run} in~$𝒯$ is an infinite
sequence of consecutive transitions,
\begin{displaymath}
   q_0 \trans{a_1|v_1} q_1
       \trans{a_2|v_2} q_2
       \trans{a_3|v_3} q_3
       ⋯
\end{displaymath}
Its \emph{input} and \emph{output labels} are the sequences of symbols
$a_1 a_2 a_3 ⋯$ and $v_1v_2v_3⋯$ respectively.

If $𝒯$ is a unambiguous transducer, each sequence~$x$ is the input label of
at most one accepting run in~$𝒯$.  When this run does exist, its output is
denoted by~$𝒯(x)$. We say that a unambiguous transducer~$𝒯$ \emph{preserves
  normality} if for each normal word~$x$, $𝒯(x)$ is also normal.

An automaton (respectively, transducer) is said to be \emph{trim} if each
state occurs in an accepting run.  Automata and transducers are always
assumed to be trim since useless states can easily be removed.

We end this section by stating very easy but useful facts about unambiguous
automata.  If $⟨ Q,A,B, Δ,I, F⟩$ is an unambiguous automaton then each
automaton $⟨ Q,A,B, Δ, \{ q \}, F⟩$ obtained by taking state~$q$ as initial
state is also unambiguous.  Similarly, removing states or transitions from
an unambiguous automaton yields an unambiguous automaton.  Combining these
two facts gives that each strongly connected component, seen as an
automaton, of an unambiguous automaton is still an unambiguous automaton.

\subsection{Weighted Automata}

We now introduce weighted automata.  In this paper we only consider
weighted automata whose weights are rational numbers with the usual
addition and multiplication (see \cite[Chap.~III]{Sakarovitch09a} for a
complete introduction).

A \emph{weighted automaton}~$𝒜$ is a tuple $⟨ Q,B,Δ,I,F ⟩$, where $Q$ is
the state set, $B$ is the alphabet, $I:Q→ ℚ$ and $F:Q→ ℚ$ are the functions
that assign to each state an initial and a final weight and
$Δ : Q × B × Q → ℚ$ is a function that assigns to each transition a weight.

As usual, the weight of a run is the product of the weights of its
transitions times the initial weight of its first state and times the final
weight of its last state.  Furthermore, the weight of a word $w ∈ B^*$ is
the sum of the weights of all runs with label~$w$ and it is denoted
$\weight_𝒜(w)$.

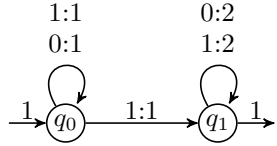
\begin{figure}[htbp]
  \begin{center}
  \begin{tikzpicture}[->,>=stealth',semithick,auto,inner sep=1.2pt]
  \tikzstyle{every state}=[minimum size=0.4]
  \node (qinitq0) at (-0.8,0) {} ;
  \node (qfinq1) at (2.8,0) {} ;
  \node[state] (q0) at (0,0) {$q_0$};
  \node[state] (q1) at (2,0) {$q_1$};
  \path (qinitq0) edge node {$1$} (q0);
  \path (q1) edge node {$1$} (qfinq1);
  \path (q0) edge node {$1{:}1$} (q1) ;
  \path (q0) edge[out=120,in=60,loop] node {$\begin{array}{c} 1{:}1 \\ 0{:}1 \end{array}$} (q0);
  \path (q1) edge[out=120,in=60,loop] node {$\begin{array}{c} 0{:}2 \\ 1{:}2 \end{array}$} (q1);
  \end{tikzpicture}
  \end{center}
  \caption{A weighted automaton}
  \label{fig:weighted} 
\end{figure}

A transition $p \trans{a} q$ with weight $x$ is pictured
$p \trans{a{:}x} q$.  Non-zero initial and final weights are given over
small incoming and outgoing arrows. A weighted automaton is pictured in
Figure~\ref{fig:weighted}. The weight of the run
$q_0 \trans{1} q_1 \trans{0} q_1 \trans{1} q_1 \trans{0} q_1$ is
$1 ⋅ 1 ⋅ 2 ⋅ 2 ⋅ 2 ⋅ 1 = 8$.  The weight of the word
$w = 1010$ is $8 + 2 = 10$.  More generally the weight of a
word~$w = a_1 ⋯ a_k$ is the integer~$n = ∑_{i=1}^k {a_i2^{k-i}}$
($w$ is a binary expansion of~$n$ with possibly some leading zeros).

A weighted automaton can also be represented by a triple $\tuple{π,μ,ν}$
where $π$ is a raw vector over~$ℚ$ of dimension $1 × n$, $μ$ is a morphism
from~$B^*$ into the set of $n × n$-matrices over~$ℚ$ and $ν$ is a column
vector of dimension $n × 1$ over~$ℚ$.  The weight of a word~$w ∈ B^*$ is
then equal to $πμ(w)ν$.  The vector~$π$ is the vector of initial weights,
the vector~$ν$ is the vector of final weights and, for each symbol~$b$,
$μ(b)$ is the matrix whose $(p,q)$-entry is the weight~$x$ of the
transition $p \trans{b:x} q$.  The weighted automaton pictured in
Figure~\ref{fig:weighted} is, for instance, represented by $\tuple{π,μ,ν}$
where
$π = (1,0)$, $ν = \left(\begin{smallmatrix} 0 \\
    1 \end{smallmatrix}\right)$ and the morphism~$μ$ is given by
\begin{displaymath}
  μ(0) =
  \left(
    \begin{array}{cc}
      1 & 0 \\
      0 & 2
    \end{array}
  \right)
  \quad\text{and}\quad
  μ(1) = 
  \left(
    \begin{array}{cc}
      1 & 1 \\
      0 & 2
    \end{array}
  \right).
\end{displaymath}

\section{Results} \label{sec:results}

We now state the main results of the paper.  The first one states that when
a transducer is strongly connected, unambiguous and complete, the frequency
of each finite word~$w$ in the output of a run with a normal input label is
given by a weighted automaton over~$ℚ$.  The second one states that it can
be checked whether an unambiguous transducer preserves normality.

\begin{theorem} \label{thm:weighted}
  Given an unambiguous and strongly connected transducer, there exists a
  weighted automaton~$𝒜$ such that for each normal sequence~$x$ in the
  domain of~$𝒯$ and for any finite word $w$, $\freq(𝒯(x),w)$ is equal to
  $\weight_𝒜(w)$.
  
  Furthermore, the weighted automaton~$𝒜$ can be computed in cubic
  time with respect to the size of the transducer~$𝒯$.
\end{theorem}

Theorem~\ref{thm:weighted} only deals with strongly connected transducers,
but Proposition~\ref{pro:decomposition} deals with the general case by
showing that it suffices to apply Theorem~\ref{thm:weighted} to some
strongly connected components to check preservation of normality.

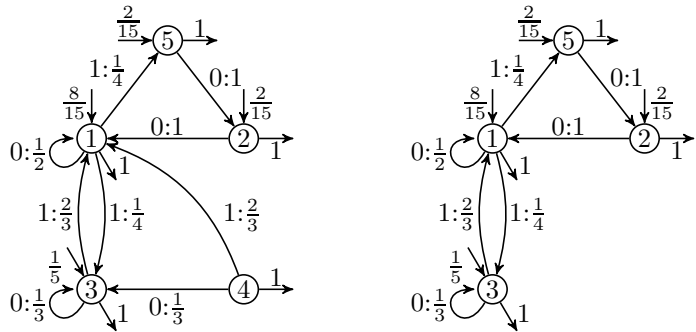
\begin{figure}[htbp]
  \begin{center}
  \begin{tikzpicture}[->,>=stealth',semithick,auto,inner sep=1pt]
    \begin{scope}
    \tikzstyle{every state}=[minimum size=0.4]
    \node[state] (q1) at (0,2) {$1$};
    \node[state] (q2) at (2,2) {$2$};
    \node[state] (q3) at (0,0) {$3$};
    \node[state] (q4) at (2,0) {$4$};
    \node[state] (q5) at (1,3.3) {$5$};
    \node (qini1) at (0,2.7) {};
    \node (qini2) at (2,2.7) {};
    \node (qini3) at (-0.35,0.6) {};
    \node (qini4) at (2,0.7) {};
    \node (qini5) at (0.3,3.3) {};
    \node (qfin1) at (0.35,1.4) {};
    \node (qfin2) at (2.7,2) {};
    \node (qfin3) at (0.35,-0.6) {};
    \node (qfin4) at (2.7,0) {};
    \node (qfin5) at (1.7,3.3) {};
    \path (qini1) edge node[swap,pos=0.4] {$\frac{8}{15}$} (q1);
    \path (qini2) edge node[pos=0.4] {$\frac{2}{15}$} (q2);
    \path (qini3) edge node[swap,pos=0] {$\frac{1}{5}$} (q3);
    \path (qini5) edge node[pos=0.3] {$\frac{2}{15}$} (q5);
    \path (q1) edge node[pos=0.9] {$1$} (qfin1);
    \path (q2) edge node[swap] {$1$} (qfin2);
    \path (q3) edge node[pos=0.8] {$1$} (qfin3);
    \path (q4) edge node {$1$} (qfin4);
    \path (q5) edge node {$1$} (qfin5);
    \path (q1) edge[out=240,in=180,loop] node[pos=0.62] {$0{:}\frac{1}{2}$} (q1);
    \path (q1) edge node {$1{:}\frac{1}{4}$} (q5);
    \path (q5) edge node {$0{:}1$} (q2);
    \path (q2) edge node[swap] {$0{:}1$} (q1);
    \path (q1) edge[bend left=15] node {$1{:}\frac{1}{4}$} (q3);
    \path (q3) edge[bend left=15] node {$1{:}\frac{2}{3}$} (q1);
    \path (q3) edge[out=240,in=180,loop] node[pos=0.62] {$0{:}\frac{1}{3}$} (q1);
    \path (q4) edge node {$0{:}\frac{1}{3}$} (q3);
    \path (q4) edge[bend right=25] node[swap,pos=0.22] {$1{:}\frac{2}{3}$} (q1);
    \end{scope}
    \begin{scope}[xshift=150]
    \tikzstyle{every state}=[minimum size=0.4]
    \node[state] (q1) at (0,2) {$1$};
    \node[state] (q2) at (2,2) {$2$};
    \node[state] (q3) at (0,0) {$3$};
    \node[state] (q5) at (1,3.3) {$5$};
    \node (qini1) at (0,2.7) {};
    \node (qini2) at (2,2.7) {};
    \node (qini3) at (-0.35,0.6) {};
    \node (qini5) at (0.3,3.3) {};
    \node (qfin1) at (0.35,1.4) {};
    \node (qfin2) at (2.7,2) {};
    \node (qfin3) at (0.35,-0.6) {};
    \node (qfin5) at (1.7,3.3) {};
    \path (qini1) edge node[swap,pos=0.4] {$\frac{8}{15}$} (q1);
    \path (qini2) edge node[pos=0.4] {$\frac{2}{15}$} (q2);
    \path (qini3) edge node[swap,pos=0] {$\frac{1}{5}$} (q3);
    \path (qini5) edge node[pos=0.3] {$\frac{2}{15}$} (q5);
    \path (q1) edge node[pos=0.9] {$1$} (qfin1);
    \path (q2) edge node[swap] {$1$} (qfin2);
    \path (q3) edge node[pos=0.8] {$1$} (qfin3);
    \path (q5) edge node {$1$} (qfin5);
    \path (q1) edge[out=240,in=180,loop] node[pos=0.62] {$0{:}\frac{1}{2}$} (q1);
    \path (q1) edge node {$1{:}\frac{1}{4}$} (q5);
    \path (q5) edge node {$0{:}1$} (q2);
    \path (q2) edge node[swap] {$0{:}1$} (q1);
    \path (q1) edge[bend left=15] node {$1{:}\frac{1}{4}$} (q3);
    \path (q3) edge[bend left=15] node {$1{:}\frac{2}{3}$} (q1);
    \path (q3) edge[out=240,in=180,loop] node[pos=0.62] {$0{:}\frac{1}{3}$} (q1);
    \end{scope}
  \end{tikzpicture}
  \end{center}
  \caption{Two weighted automata}
  \label{fig:weighted2}
\end{figure}

To illustrate Theorem~\ref{thm:weighted} we give in
Figure~\ref{fig:weighted2} two weighted automata which compute the
frequency of each finite word~$w$ in~$𝒯(x)$ for a normal input~$x$ and the
transducer~$𝒯$ pictured in Figure~\ref{fig:transducer}.  The leftmost one
is obtained by the procedure described in the next section. The rightmost
one is obtained by removing useless states from the leftmost one.

\begin{theorem} \label{thm:preservation}
  It can be decided in cubic time whether an unambiguous transducer
  preserves normality or not.
\end{theorem}

From the weighted automaton pictured in Figure~\ref{fig:weighted2},
it is easily computed that the limiting frequencies of the digits $0$
and~$1$ in the output~$𝒯(x)$ of a normal input~$x$ are respectively
$9/15$ and $6/15$.  This shows that the transducer~$𝒯$ pictured in
Figure~\ref{fig:transducer} does not preserve normality.

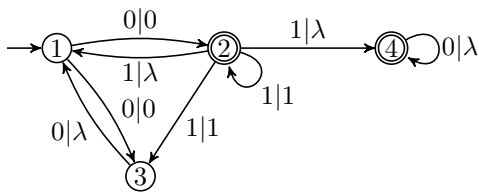
\begin{figure}[htbp]
  \begin{center}
  \begin{tikzpicture}[->,>=stealth',initial text=,semithick,auto,inner sep=1pt]
  \tikzstyle{every state}=[minimum size=0.4]
  \node[state,initial left] (q1) at (0,1.7) {$1$};
  \node[state,accepting] (q2) at (2.2,1.7) {$2$};
  \node[state] (q3) at (1.1,0) {$3$};
  \node[state,accepting] (q4) at (4.4,1.7) {$4$};
  \path (q1) edge[bend left=10] node {$0|0$} (q2);
  \path (q2) edge[out=-15,in=-75,loop] node {$1|1$} (q2);
  \path (q2) edge[bend left=10] node {$1|\emptyword$} (q1);
  \path (q1) edge[bend left=10] node[pos=0.7] {$0|0$} (q3);
  \path (q3) edge[bend left=10] node {$0|\emptyword$} (q1);
  \path (q2) edge node {$1|\emptyword$} (q4);
  \path (q2) edge node {$1|1$} (q3);
  \path (q4) edge[out=30,in=-30,loop] node {$0|\emptyword$} (q4);
  \end{tikzpicture}
  \end{center}
  \caption{Another unambiguous transducer}
  \label{fig:transducer2}
\end{figure}

To illustrate the previous theorem, we show that the transducer
pictured in Figure~\ref{fig:transducer2} is unambiguous and does preserve
normality.  It is actually a selector as defined below in
Section~\ref{sec:selection} because the output of each transition is either
the input symbol or the empty word.  Therefore, the output is always a
subsequence of the input sequence.  It can be checked that a symbol is
selected, that is copied to the output, if the number of~$0$ until the
next~$1$ is finite and even, including zero.

By Proposition~\ref{pro:decomposition} below, it suffices to check that the
strongly connected component made of the states $\{ 1, 2, 3\}$ does
preserve normality.  The weighted automaton given by the algorithm is
represented by the triple $\tuple{π,μ,𝟏}$ where $π$ is the raw vector
$π = (3/4,1/4)$, $𝟏$ is the column vector
$\left(\begin{smallmatrix} 1 \\ 1 \end{smallmatrix}\right)$ and the
morphism~$μ$ is defined by
\begin{displaymath}
  μ(0) =
  \left(
    \begin{array}{cc}
      1/4 & 1/12 \\
      3/4 & 1/4
    \end{array}
  \right)
  \quad\text{and}\quad
  μ(1) = 
  \left(
    \begin{array}{cc}
      1/2 & 1/6 \\
      0 & 0
    \end{array}
  \right).
\end{displaymath}
The vector~$π$ satisfies $πμ(0) = πμ(1) = \frac{1}{2}π$ and therefore
$πμ(w)𝟏$ is equal to $2^{-|w|}$ for each word~$w$.  This shows that the
transducer pictured in Figure~\ref{fig:transducer2} does preserve
normality.

\section{Adjacency matrix of the automaton}

In this section, we introduce the adjacency matrix of an automaton.  This
matrix is particularly useful when the automaton is strongly connected and
unambiguous. Its spectral radius characterizes the fact that the automaton
does accept or not a normal sequence as stated in
Proposition~\ref{pro:acceptnormal} below.

Let $𝒜$ be an automaton with state set~$Q$.  The \emph{adjacency matrix}
of~$𝒜$ is the $Q × Q$-matrix~$M$ defined by
$M_{p,q} = \#\{ a ∈ A: p \trans{a} q\}/\#A$.  Its entry $M_{p,q}$ is thus
the number of transitions from~$p$ to~$q$ divided by the cardinality of the
alphabet~$A$.  The factor~$1/\#A$ is just a normalization factor to compare
the spectral radius of this matrix to~$1$ rather than to the cardinality of
the alphabet.  By a slight abuse of notation, the spectral radius of the
adjacency matrix, is called the spectral radius of the automaton.

The adjacency matrix of the unambiguous automaton pictured in
Figure~\ref{fig:unambiguous} is the matrix $M$ given by
\begin{displaymath}
  M = \frac{1}{2}
  \left(
    \begin{array}{cccc}
      1 & 1 & 1 & 0 \\
      1 & 0 & 0 & 0 \\
      0 & 0 & 0 & 1 \\
      1 & 0 & 2 & 0
    \end{array}
  \right)
\end{displaymath}
It can be checked that the spectral radius of this matrix is~$1$.

We implicitly suppose that the automaton~$𝒜$ has at least one initial state
and one final state.  Otherwise, no sequence is accepted by it and nothing
interested can be said about it.  For each state~$q$, let $F_q$ be the
\emph{future} set, that is the set $F_q = \{ x : q \trans{x} \limrun \}$ of
accepted sequences if $q$ is taken as the only initial state. Let $α_q$ be
the measure of the set~$F_q$.  Note that the sum $∑_{q ∈ Q}{α_q}$ might be
greater than~$1$ because the sets~$F_q$ might not be pairwise disjoint.  We
claim that the vector $α = (α_q)_{q ∈ Q}$ satisfies $Mα = α$.  This
equality means that either $α$ is the zero vector or that $α$ is a right
eigenvector of~$M$ for the eigenvalue~$1$.  This equality comes from the
following relations between the sets~$F_q$.  For each state~$p$, one has
\begin{displaymath}
  F_p = ⨄_{p \trans{a} q} aF_q
\end{displaymath}
where the symbol~$⊎$ stands for the union of pairwise disjoint sets.
The fact that $F_p$ is equal to the union of the sets $aF_q$ for $a$ and
$q$ ranging over all possible transitions $p \trans{a} q$ is true in any
automaton accepting sequences.  Furthermore, the unambiguity of~$𝒜$ implies
that the sets~$aF_q$ for different pairs $(a,q)$ must be pairwise disjoint.
Therefore, for each state~$p$,
\begin{displaymath}
  \alpha_p = μ(F_p) = \sum_{p \trans{a} q}  μ(aF_q)
           = \frac{1}{\#Q}\sum_{p \trans{a} q}  α_q.
\end{displaymath}

By definition, the adjacency matrix is non-negative.  By the
Perron-Frobenius theorem, its spectral radius must be one of its
eigenvalues.  The following lemma states that if the automaton~$𝒜$ is
strongly connected and unambiguous, then the spectral radius of~$M$ is less
than~$1$.
\begin{lemma}
  Let $𝒜$ be a strongly connected and unambiguous automaton.  The
  maximum eigenvalue of its adjacency matrix~$M$ is less than~$1$.
\end{lemma}
\begin{proof}
  Consider the matrix $\#A ⋅ M$.  Its $(p,q)$-entry is the number of
  transitions from~$p$ to~$q$.  It follows that the $(p,q)$-entry of
  $(\#A ⋅ M)^n$ is the number of runs of length~$n$ from~$p$ to~$q$.  Since
  $𝒜$ is unambiguous, each finite word is the label of at most one run
  from~$p$ to~$q$.  This yields that the entry $(\#A ⋅ M)^n_{p,q}$ is
  bounded by the number $(\#A)^n$ of words of length~$n$ and that each
  entry of~$M^n$ is bounded by~$1$.

  Since the automaton~$𝒜$ is strongly connected, the matrix~$M$ is positive
  and irreducible.  Let $λ$ be its spectral radius which is a positive real
  number.  Suppose that the period of~$M$ is the positive integer~$p$.  By
  Theorem~1.4 in~\cite{Senata06}, The matrix $M^p$ can be decomposed as
  diagonal blocks of primitive matrices and at least one of this block~$M'$
  has $λ^p$ as eigenvalue.  For a positive matrix~$M'$, there exists, by
  Theorem~1.2 in~\cite{Senata06}, a constant~$K$ such that each entry of
  the matrix~$M^{\prime n}$ satisfies $λ^n/K ⩽ M^{\prime n}_{p,q} ⩽ Kλ^n$.
  This proves that $λ ⩽ 1$.
\end{proof}

By the previous lemma, the spectral radius of the adjacency matrix of an
unambiguous automaton is less than~$1$.  The following proposition states
when it is equal to~$1$ or strictly less than~$1$.
\begin{proposition} \label{pro:acceptnormal}
  Let $𝒜$ be a strongly connected and unambiguous automaton and let $λ$ be
  the spectral radius of its adjacency matrix. If $λ = 1$ then $𝒜$ accepts
  at least one normal sequence and each number~$α_q$ is positive.  If
  $λ < 1$, then $𝒜$ accepts no normal sequence and each number~$α_q$ is
  equal to zero.
\end{proposition}
\begin{proof}
  Let $M$ be the adjacency matrix of~$𝒜$.  Suppose first that its spectral
  radius satisfies $λ < 1$.  The number of runs of length~$n$ is equal to
  the sum $(\#A)^n∑_{p,q ∈ Q} M^n_{p,q}$ where $M^n_{p,q}$ is the
  $(p,q)$-entry of~$M^n$.  By the Perron-Frobenius theorem, there exists a
  constant~$K$ such that $M^n_{p,q} ⩽ Kλ^n$ for each $p,q ∈ Q$. It follows
  that the number of words which are the label of some run in~$𝒜$ is less
  than $K(\#Q)^2(λ\#A)^n$.  For $n$ great enough, $K(\#Q)^2λ^n$ is strictly
  less than~$1$ and some word of length~$n$ is the label of no run in~$𝒜$.
  This implies that no normal sequence can be accepted by~$𝒜$.  The
  equality $Mα = α$ shows that $α = 0$ since $1$ is not an eigenvalue
  of~$M$.

  We now suppose that the spectral radius~$λ$ of~$M$ is~$1$.  The entropy
  of the sofic shift defined by~$𝒜$ is $\log_2 \#A$, each finite word is
  the label of at least one run in~$𝒜$.  Otherwise, the entropy of the
  sofic shift would be strictly less than~$\log_2 \#A$.  Let
  $x = a_1a_2a_3 ⋯ $ be a normal sequence.  Each prefix $a_1 ⋯ a_n$ of~$x$
  is the label of a run in~$𝒜$.  By extraction, we get a run whose label
  is~$x$.  Note that this run might be neither initial nor final.  To get
  an initial run we consider the new sequence~$ux$ where $u$ is the label
  of a run from an initial state to the starting state of the run labelled
  by~$x$.  To get a final run, we insert in~$x$ at positions of the
  form~$2^k$ a word of length at most $2\#Q$ to make a small detour to a
  final state of~$𝒜$.  Since the inserted blocks have bounded lengths and
  they are inserted at sparse positions, the obtained sequence is still
  normal.  We now prove that $α$ is positive.  We claim that almost all
  sequences are the label of a run visiting infinitely often each state.
  It suffices to prove that for each state~$p$, almost all sequences are
  the label of a run visiting infinitely often~$p$.  Let $𝒜'$ be the
  automaton obtained by removing all transitions starting from~$p$.  By
  Theorem~1.5e in~\cite{Senata06}, the spectral radius of the adjacency
  matrix of~$𝒜'$ is strictly less than~$1$.  Therefore, by the previous
  case, the measures~$α'_q$ of the sets
  $F'_q = \{ x : q \trans{x} ∞ \text{ in }𝒜'\}$ are equal to zero.  This
  proves that the set of sequences which are the label of a run never
  visiting~$p$ has measure~$0$.  By the same reasoning, it can be shown
  that the set of sequences which are the label of a run visiting~$p$
  finitely many times has also measure~$0$.  This proves the claim.  This
  shows that at least one entry of~$α$ must be positive. Since $Mα = α$ and
  $M$ is irreducible, all entries of~$α$ are positive.
\end{proof}

The spectral radius of the adjacency matrix of the automaton pictured in
Figure~\ref{fig:unambiguous} is~$1$.  The vector~$α$ is given by
$α_1 = α_4 = 2/3$ and $α_2 = α_3 = 1/3$.

Suppose that the automaton~$𝒜$ is strongly connected and unambiguous and
that the spectral radius of its adjacency matrix~$M$ is~$1$.  The
vector~$α$ is then a right eigenvector of~$M$ for the eigenvalue~$1$.
There is also a left eigenvector $π = (π_q)_{q∈ Q}$ for the eigenvalue~$1$.
This vector~$π$ is strictly positive and we normalize it in such a way that
$∑_{q ∈ Q}{π_qα_q} = 1$.

The left eigenvector~$π$ of the adjacency matrix of the automaton pictured
in Figure~\ref{fig:unambiguous}, normalized as explained just above, is
given by $π_1 = π_3 = 2/3$ and $π_2 = π_4 = 1/3$.  The vector of
$(π_qα_q)_{q ∈ Q}$ is then given $π_1α_1 = 4/9$, $π_2α_2 = π_3α_3 = 2/9$
and $π_4α_4 = 1/9$.

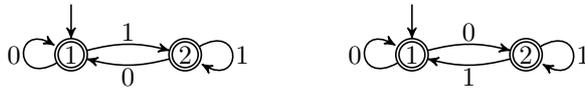
\begin{figure}[htbp]
  \begin{center}
  \begin{tikzpicture}[->,>=stealth',initial text=,semithick,auto,inner sep=1pt]
  \tikzstyle{every state}=[minimum size=0.4]
  \node[state,initial above,accepting] (q1) at (0,0) {$1$};
  \node[state,accepting] (q2) at (1.5,0) {$2$};
  \path (q1) edge[out=210,in=150,loop] node {$0$} (q1);
  \path (q1) edge [bend left=15] node {$1$} (q2);
  \path (q2) edge [bend left=15] node {$0$} (q1);
  \path (q2) edge[out=30,in=-30,loop] node {$1$} (q2);
  \end{tikzpicture}
  \hspace{1cm}
  \begin{tikzpicture}[->,>=stealth',initial text=,semithick,auto,inner sep=1pt]
  \tikzstyle{every state}=[minimum size=0.4]
  \node[state,initial above,accepting] (q1) at (0,0) {$1$};
  \node[state,accepting] (q2) at (1.5,0) {$2$};
  \path (q1) edge[out=210,in=150,loop] node {$0$} (q1);
  \path (q1) edge [bend left=15] node {$0$} (q2);
  \path (q2) edge [bend left=15] node {$1$} (q1);
  \path (q2) edge[out=30,in=-30,loop] node {$1$} (q2);
  \end{tikzpicture}
  \end{center}
  \caption{Unambiguous automata}
  \label{fig:matrix}
\end{figure}

We would like to emphasize that that the adjacency matrix~$M$ is not
sufficient to compute the vector~$α$.  Said differently, two automata with
the same adjacency matrix may have different vectors~$α$.  Consider the two
automata pictured in Figure~\ref{fig:matrix}.  The leftmost one is
deterministic whereas the rightmost one is reverse deterministic.  Both
automata have the same matrix
$M = \frac{1}{2}\left(\begin{smallmatrix} 1&1 \\
    1&1 \end{smallmatrix}\right)$ as adjacency matrix.  For the leftmost
automaton the sets $F_1$ and~$F_2$ are both equal to $\{ 0, 1\}^ℕ$ and thus
$α_1 = α_2 = 1$.  For the leftmost automaton the sets $F_1$ and~$F_2$ are
respectively equal to $0\{ 0, 1\}^ℕ$ and $1\{ 0, 1\}^ℕ$ and thus
$α_1 = α_2 = 1/2$.

Note however that since $α$ is the eigenvector of the
irreducible matrix~$M$ for its Perron-Frobenius eigenvalue, it is unique up
to a multiplicative factor.  This means that the ratios $α_q/α_p$ can be
computed from the matrix~$M$.
     
\section{Markov chain of an unambiguous automaton}

In this section, we introduce a Markov chain associated with an unambiguous
automaton.  The use of the ergodic theorem applied to this Markov chain is
the main ingredient in the proof of Theorem~\ref{thm:weighted}.  Let $𝒜$ be
a strongly connected and unambiguous automaton and let $p$ be one of its
states.  We also suppose that the spectral radius of its adjacency
matrix~$M$ is~$1$.  By Proposition~\ref{pro:acceptnormal}, the
measure~$α_q$ of each set~$F_q = \{ x : q \trans{x} ∞ \}$ is non-zero.

We define a stochastic process $(X_n)_{n ⩾ 0}$ as follows.  Its sample set
is the set $F_p ⊆ A^ℕ$ equipped with the uniform measure~$μ$. For each
sequence~$x = x_1x_2x_3 ⋯$ in~$F_p$, there exists a unique accepting run
\begin{displaymath}
  p = q_0 \trans{x_1} q_1 \trans{x_2} q_2 \trans{x_3} q_3 ⋯ 
\end{displaymath}
The process is defined by setting $X_n(x) = q_n$ for each $x ∈ F_p$.
The following proposition states the main property of this process.
\begin{proposition}
  The process $(X_n)_{n ⩾ 0}$ is a Markov chain.
\end{proposition}
\begin{proof}
  We prove that this process is actually a Markov chain.  A sequence~$x$
  satisfies $X_n(x) = q_n$ if and only if, when factorizing $x$ as $x = wy$
  with $w = x[1{:}n]$, the word~$w$ is the label of a finite run
  $p \trans{w} q_n$ and the sequence~$y$ belongs to the set~$F_q$.  This
  remark allows us to compute the probability that $X_n = q_n$ for a given
  state~$q_n$.
  \begin{align*}
    \prob(X_n = q_n) & = μ(\{ w ∈ A^n : p \trans{w} q_n\}) α_{q_n} \\
                     & = \#\{ w ∈ A^n : p \trans{w} q_n\} α_{q_n}/(\#A)^n 
  \end{align*}
  A similar reasoning allows us to compute the probability that $X_n = q_n$
  and $X_{n+1} = q_{n+1}$ for two given states $q_n$ and~$q_{n+1}$.
  \begin{displaymath}
    \prob(X_{n+1} = q_{n+1}, X_n = q_n)
    = μ(\{ wa ∈ A^{n+1} : p \trans{w} q_n \trans{a} q_{n+1}\}) α_{q_{n+1}} 
  \end{displaymath}
  Using the definition of conditional probability, we get 
  \begin{displaymath}
    \prob(X_{n+1} = q_{n+1} | X_n = q_n) =
    \frac{\#\{ a ∈ A : q_n \trans{a} q_{n+1}\}α_{q_{n+1}}}{(\#A)α_{q_n}}
  \end{displaymath}
  Let $q_0,…, q_n$ be $n+1$ states of the automaton~$𝒜$ such that
  $q_0 = p$.  A sequence~$x$ satisfies $X_n(x) = q_n,… , X_0(x) = q_0$ if
  and only if the sequence~$x$ can be factorized $x = wx'$ where the
  word~$w = a_1 ⋯ a_n$ is the prefix of length~$n$ of~$x$, there is a
  finite run $q_0 \trans{a_1} q_1 ⋯ q_{n-1} \trans{a_n} q_n$ and $x'$
  belongs to the set~$F_{q_n}$.
  \begin{displaymath}
    \prob(X_n = q_n,…,X_0 = q_0)
    = μ(\{ a_1⋯ a_n  ∈ A^n : q_0  \trans{a_1} q_1 ⋯
                                   q_{n-1} \trans{a_n} q_n\}) α_{q_n} 
  \end{displaymath}
  Using again the definition of conditional probability, we get 
  \begin{displaymath}
    \prob(X_{n+1} = q_{n+1} | X_n = q_n,\dots,X_0 = q_0) =
      \frac{\#\{ a ∈ A : q_n \trans{a} q_{n+1}\} α_{q_{n+1}}}{(\#A)α_{q_n}}
  \end{displaymath}
  Since $\prob(X_{n+1} = q_{n+1} | X_n = q_n)$ and
  $\prob(X_{n+1} = q_{n+1} | X_n = q_n,\dots,X_0 = q_0)$ have the same
  value, the process is indeed a Markov chain.
\end{proof}
Let $P$ the $Q × Q$-matrix of probabilities for the introduced Markov
chain.  For each states $p,q ∈ Q$, the $(p,q)$-entry of~$P$ is given by
$P_{p,q} = \#\{ a ∈ A : p \trans{a} q \}α_q/(\#A)α_p$.  Note that the
matrix~$P$ and the adjacency matrix~$M$ of~$𝒜$ are related by the
equalities $P_{p,q} = M_{p,q}α_q/α_p$ for each states $p,q ∈ Q$.  We claim
that the stationary distribution of the stochastic matrix~$P$ is the
vector $(π_qα_q)_{q ∈ Q}$ where $π = (π_q)_{q ∈ Q}$ is the left eigenvector
of the matrix~$M$ for the eigenvalue~$1$.  Let us recall that $π$ has been
normalized such that $∑_{q ∈ Q} π_qα_q = 1$.
\begin{align*}
  ∑_{p ∈ Q} π_pα_pP_{p,q} & = α_q ∑_{p ∈ Q}π_pM_{p,q} \\
                          & = π_qα_q
\end{align*}

Runs are defined as sequences of consecutive transitions, and can be
considered as words over the alphabet made of all transitions.  Therefore,
the notion of \emph{frequency} $\freq(ρ, γ)$ of a finite
run~$γ$ in an infinite run~$ρ$ is defined as in
Section~\ref{sec:basic}.  Note that $\freq(ρ, γ)$ is a limit and
might not exist.  As a run can merely be regarded as a sequence of states,
$\freq(ρ,q)$ is defined similarly when $q$ is a state.  Note that
$\freq(ρ,q)$ could equivalently be defined as the sum of all
$\freq(ρ,τ)$ where $τ$ ranges over all transitions (seen as runs
of length~$1$) starting from~$q$.

The application of the ergodic theorem to the previous Markov chain is used
to prove the following proposition.  It states that in a run whose label is
a normal sequence, each state has a limiting frequency and that this
frequency is given by the stationary distribution $(π_qα_q)_{q ∈ Q}$.  This
statement is an extension to unambiguous automata of Lemma~4.5
in~\cite{Schnorr71} which is only stated for deterministic automata.

\begin{proposition} \label{pro:statefreq}
  Let $𝒜$ be a strongly connected and unambiguous automaton such that the
  spectral radius of its adjacency matrix is~$1$.  Let $ρ$ be an accepting
  run whose label is a normal sequence.  Then, for any state~$r$
  \begin{displaymath}
    \lim_{n→∞}\frac{\occ{ρ[1{:}n]}{r}}{n} = π_rα_r
  \end{displaymath}
  where $ρ[1{:}n]$ is the finite run made of the first $n$ transitions
  of~$ρ$.
\end{proposition}

Note that the result of Proposition~\ref{pro:statefreq} implies that the
frequencies of states do not depend on the input as long as this input is
normal.  Note also that the result is void if the spectral radius of the
adjacency matrix is less than~$1$ because, by
Proposition~\ref{pro:acceptnormal}, no accepting run is labeled by a normal
sequence.  This assumption could be removed because the statement remains
true but this is our choice to mention explicitly the assumption for
clarity.

The proof of the proposition is based on the following lemma.  Since the
automaton~$𝒜$ in unambiguous, there is, for two given states $p$ and~$q$
and a given word~$w$, a unique run from~$p$ to~$q$ labeled by~$w$.  This
run is written $p \trans{w} q$ as usual.

\begin{lemma} \label{lem:ergodic}
  Let $𝒜$ be a strongly connected and unambiguous automaton such that the
  spectral radius of its adjacency matrix is~$1$.  Let $r$ be a fixed state
  of~$𝒜$.  For any positive real numbers $δ, ε > 0$, there exists an
  integer~$n$ such that for each integer $k ⩾ n$,
  \begin{displaymath}
    \#\{ w ∈ A^k : ∃ p,q ∈ Q^2 \;\;
    |\occ{p \trans{w} q}{r}/k - π_rα_r| > δ \} < ε(\#A)^k
  \end{displaymath}
\end{lemma}
Note that for a triple $(p,w,q)$, the finite run $p \trans{w} q$ might not
exist.  When we write $|\occ{p \trans{w} q}{r}/k - π_rα_r| > δ$, it should
be understood as follows. The run $p \trans{w} q$ does exist and it
satisfies $|\occ{p \trans{w} q}{r}/k - π_rα_r| > δ$.

\begin{proof}
  Since there are finitely many pairs $(p,q)$ in~$Q^2$, it suffices to
  prove
  \begin{displaymath}
    \#\{ w ∈ A^k : |\occ{p \trans{w} q}{r}/k - π_rα_r| > δ \} < ε(\#A)^k/(\#Q)^2
  \end{displaymath}
  for each pair~$(p,q)$ in~$Q^2$.  Therefore, we assume that a pair $(p,q)$
  is fixed.  We consider the Markov chain $(X_n)_{n⩾0}$ introduced above
  with initial state~$p$.  This Markov chain is irreducible because $𝒜$ is
  strongly connected.  We apply the ergodic theorem for Markov chains
  \cite[Thm~4.1]{Bremaud08} to the function $f = 𝟙_r$ defined by
  $𝟙_r(s) = 1$ if $s = r$ and $𝟙_r(s) = 0$ otherwise.  It follows that
  $\lim_{n→∞} S_n = π_rα_r$ for almost all sequences in~$F_p$ where
  $S_n = \frac{1}{n}∑_{i=1}^n 𝟙_r(X_i)$.  The positive numbers $δ$ and~$ε$
  being fixed, there is an integer~$n$ such that, for each $k ⩾ n$, the
  measure of the set $\{ x : |S_k - π_rα_r| > δ \}$ is less than~$ε$.
  Consider now a set~$F$ of the form $F = wF_q$ where $w$ is a word of
  length~$k ⩾ n$ and $q$ is the state that has been fixed.  The measure
  of~$F$ is $α_q / (\#A)^k$. If the run $p \trans{w} q$ does exist, then
  $S_k$ is constant on the set~$F$ because the number of occurrences of~$r$
  in the first $k$ positions of the run only depends on the finite run
  $p \trans{w} q$.  It follows that the number of words~$w$ of length~$k$
  such that $p \trans{w} q$ does exist and satisfies
  $|\occ{p \trans{w} q}{r}/k = π_rα_r| > δ$ is bounded by $ε(\#A)^k/α_q$.
  The result is then obtained by replacing $ε$ by $ε\min \{ α_q : q ∈ Q\}$
  which is positive because all entries of~$α$ are positive.
\end{proof}

We now come to the proof of Proposition~\ref{pro:statefreq}.
\begin{proof}[Proof of Proposition~\ref{pro:statefreq}]
  Let $ρ$ be an accepting run in~$𝒜$ whose label is a normal sequence~$x$.
  Since $∑_{q ∈ Q} π_qα_q = 1$, it suffices to prove that
  $\liminf_{n→∞}\occ{ρ[1{:}n]}{r}/n ⩾ π_rα_r$ for each state~$r$.  We fix
  an arbitrary positive real number~$ε$.  Applying Lemma~\ref{lem:ergodic}
  with $δ = ε$, we get an integer~$k$ such that the set $B ⊆ A^k$ defined
  by
  \begin{displaymath}
    B = \{ w ∈ A^k : ∃ p,q ∈ Q^2 \;\;
    |\occ{p \trans{w} q}{r}/k - π_rα_r| > ε \} 
  \end{displaymath}
  satisfies $\#B < ε(\#A)^k$.  The run~$ρ$ is then factorised
  \begin{displaymath}
    ρ = q_0 \trans{w_1} q_1 \trans{w_2} q_2 \trans {w_3} q_3 ⋯ 
  \end{displaymath}
  where each word~$w_i$ has length~$k$ and $x = w_1w_2w_3⋯$ is a
  factorization of~$x$ in blocks of length~$k$. Since the sequence~$x$ is
  normal, there is an integer~$N$ such that for each $n ⩾ N$ and each
  word~$w$ of length~$k$, the cardinality of the set
  $\{ 1 ⩽ i ⩽ n : w_i = w \}$ satisfies
  $\#\{ 1 ⩽ i ⩽ n : w_i = w \} ⩾ n(1 - ε)/(\#A)^k$.
  \begin{align*}
    \liminf_{n→∞}\frac{\occ{ρ[1{:}n]}{r}}{n} 
       & = \liminf_{n→∞} \frac{1}{nk}
         ∑_{i = 1}^n \occ{q_{i-1} \trans{w_i} q_i}{r} \\
       & ⩾ \liminf_{n→∞} 
         ∑_{w ∈ A^k} \frac{\#\{ 1 ⩽ i ⩽ n : w_i = w \}}{n}
         \min_{p,q ∈ Q^2} \frac{\occ{p \trans{w} q}{r}}{k} \\
       & ⩾ ∑_{w ∉ B} (1 - ε)(π_rα_r-ε)/(\#A)^k \\
       & ⩾ (1 - ε)^2(π_rα_r-ε)
  \end{align*}
  Since this is true for any $ε > 0$,
  $\liminf_{n→∞}\occ{ρ[1{:}n]}{r}/n ⩾ π_rα_r$ holds for each
  state~$r$. This completes the proof of the proposition.
\end{proof}

Proposition~\ref{pro:statefreq} states that each state has a frequency in a
run whose label is normal.  This result can be extended to finite runs as
follows.  The stochastic matrix~$P$ and its stationary distribution
$(π_qα_q)_{q ∈ Q}$ induce a canonical distribution on finite runs in the
automaton~$𝒜$.  This distribution is defined as follows.
\begin{displaymath}
  \prob(q_0 \trans{a_1} q_1 ⋯ q_{n-1} \trans{a_n} q_n) =
  \frac{π_{q_0} α_{q_n}}{(\#A)^n}.
\end{displaymath}
We claim that for each integer~$n$, this is indeed a distribution
on runs of length~$n$.  This means that 
\begin{displaymath}
  ∑_{p,q ∈ Q^2, w ∈ A^n} \prob(p \trans{w} q) = 1.
\end{displaymath}
For each symbol $a ∈ A$, define the $Q × Q$-matrix~$P_a$ by
\begin{displaymath}
  (P_a)_{p,q} = 
  \begin{cases}
    \frac{α_q}{(\#A)α_p} & \text{if $p \trans{a} q$ is a transition of $𝒜$} \\
    0                    & \text{otherwise}.
  \end{cases}
\end{displaymath}
Note that the stochastic matrix~$P$ is equal to the sum $∑_{a ∈ A}P_a$.
Let $w = a_1 ⋯ a_n$ be a word of length~$n$. Let us write $P_w$ for the
product $P_{a_1} ⋯ P_{a_n}$.  This notation is consistent with the notation
$P_a$ for words of length~$1$.  The $(p,q)$-entry of the matrix
$P_w = P_{a_1} ⋯ P_{a_n}$ is equal to
$α_q/(\#A)^nα_p = \prob(p \trans{w} q)/π_pα_p$ if the run $p \trans{w} q$
does exist and to~$0$ otherwise.
\begin{align*}
  ∑_{p,q ∈ Q^2, w ∈ A^n} \prob(p \trans{w} q)
  & = ∑_{p,q ∈ Q^2} π_pα_p ∑_{w ∈ A^n} (P_w)_{p,q} \\
  & = ∑_{p,q ∈ Q^2} π_pα_p (P^n)_{p,q} \\
  & = ∑_{p ∈ Q} π_pα_p = 1
\end{align*}

The following proposition extends to finite runs the statement of
Proposition~\ref{pro:statefreq} about states.
\begin{proposition} \label{pro:runfreq}
  Let $𝒜$ be a strongly connected and unambiguous automaton such that the
  spectral radius of its adjacency matrix is~$1$.  Let $ρ$ be an accepting
  run whose label is a normal sequence.  For any finite run $γ = q_0
  \trans{a_1} q_1 ⋯ q_{n-1} \trans{a_n} q_n$ of length~$n$, one has
  \begin{displaymath}
    \lim_{n→∞}\frac{\occ{ρ[1{:}n]}{γ}}{n} =
    \prob(γ) = \frac{π_{q_0}α_{q_n}}{(\#A)^n}
  \end{displaymath}
  where $ρ[1{:}n]$ is the finite run made of the first $n$ transitions
  of~$ρ$.
\end{proposition}
\begin{proof}
  We now define an automaton whose states are the runs of length~$n$
  in~$𝒜$.  We let $𝒜^n$ denote the automaton whose state set is
  $\{ p \trans{w} q : p,q ∈ Q, w ∈ A^n\}$ and whose set of transitions is
  defined by
  \begin{displaymath}
    \left\{
       γ  \trans{a} γ' :
       \begin{array}{c}
         γ = p \trans{b} p' \trans{w} q \\
         γ' = p' \trans{w} q \trans{a} q'
       \end{array}
       \text{with } a,b ∈ A \text{ and } w \in A^{n-1}
    \right\}
  \end{displaymath}
  The Markov chain associated with the automaton~$𝒜^n$ is called the
  \emph{snake} Markov chain.  See Problems 2.2.4, 2.4.6 and 2.5.2 (page~90)
  in \cite{Bremaud08} for more details.  It is pure routine to check that
  the stationary distribution~$ξ$ of~$𝒜^n$ is given by
  $ξ_{p \trans{w} q} = \prob(p \trans{w} q) = π_pα_q/(\#A)^n$ for each
  finite run $p \trans{w} q$ of length~$n$ in~$𝒜$.  To prove the statement,
  apply Proposition~\ref{pro:statefreq} to the automaton~$𝒜^n$.
\end{proof}

It should be pointed out that the distribution on finite path which is
defined above is the Parry measure of the edge shift of the automaton. This
shift is the shift of finite type whose symbols are the edges of the
automaton \cite[Thm~6.2.20]{Kitchens98}.

Let $γ$ be a finite run whose first state is~$p$ and let $ρ$ be an infinite
run.  We call \emph{conditional frequency} of~$γ$ in~$ρ$ the ratio
$\freq(ρ,γ)/\freq(ρ,p)$.  It is defined as soon as both frequencies
$\freq(ρ,γ)$ and $\freq(ρ,p)$ do exist.  The corollary of Propositions
\ref{pro:statefreq} and~\ref{pro:runfreq} is the following.
\begin{corollary} \label{cor:condfreq}
  The conditional frequency of a finite run $p \trans{w} q$ of length~$n$
  in an accepting run whose label is normal is $α_q/(\#A)^nα_p$.
\end{corollary}

\section{Algorithms and Proofs} \label{sec:proofs}

In this section we provide the proofs for Theorems \ref{thm:weighted}
and~\ref{thm:preservation}.  The proofs are organized in three parts.
First, the transducer~$𝒯$ is normalized into another transducer~$𝒯'$
realizing the same function.  Then this latter transducer is used to define
a weighted automaton~$𝒜$.  Second, the proof that the construction of~$𝒜$
is correct is carried out.  Third, the algorithms computing $𝒜$ and
checking whether $𝒯$ preserves normality or not are given.

By Proposition~\ref{pro:decomposition} below, it suffices to analyze
preservation of normality in some of the strongly connected components.
\begin{proposition} \label{pro:decomposition}
  An unambiguous transducer~$𝒯$ preserves normality if and only each
  strongly connected component of~$𝒯$ with a final state and spectral
  radius~$1$ preserves normality.
\end{proposition}
\begin{proof}
  Let $ρ$ be an accepting run of~$𝒯$ whose label is a normal sequence.
  This run ends in some strongly connected~$C$, that is, all states which
  are visited infinitely often by the run belong to the same strongly
  connected component~$C$. This component~$C$ must contain a final state
  because $ρ$ is accepting and by Proposition~\ref{pro:acceptnormal}, its
  spectral radius must be one.

  Conversely, let $C$ be a strongly component with a final state and
  spectral radius~$1$. By Proposition~\ref{pro:acceptnormal}, there is a
  final run~$ρ$ contained in~$C$ and whose label is a normal sequence~$x$.
  Suppose that $ρ$ starts from state~$q$ in~$C$.  Let $u$ be the label of a
  run from an initial state to~$q$.  The unique accepting run labeled
  by~$ux$ ends in~$C$ and $C$ must preserve normality.
\end{proof}

Consider for instance the transducer pictured in
Figure~\ref{fig:transducer2}.  It has two strongly connected components:
the one made of states $1,2,3$ and the one made of state~$4$.  The
corresponding adjacency matrices are
\begin{displaymath}
  \frac{1}{2}
  \left(
    \begin{array}{ccc}
      0 & 1 & 1 \\
      1 & 1 & 1 \\
      1 & 0 & 0 
    \end{array}
  \right)
  \quad\text{and}\quad
  \left(
    \begin{array}{c}
      \frac{1}{2}
    \end{array}
  \right)
\end{displaymath}
whose spectral radii are respectively $1$ and~$1/2$.  It follows that
the transducer preserves normality if an only if the transducer reduced
to the states $1,2,3$ does preserve normality.

In what follows we only consider strongly connected transducers.
Propositions \ref{pro:statefreq} and~\ref{pro:runfreq} have the following
consequence.  Let $𝒯$ be a unambiguous and strongly connected transducer.
If each transition has an empty output label, the output of any run is
empty and then $𝒯$ does not preserve normality.  Therefore, we assume that
transducers have at least one transition with a non empty output label.  By
Propositions \ref{pro:statefreq} and~\ref{pro:runfreq}, this transition is
visited infinitely often if the input is normal because the stationary
distribution $(π_qα_q)_{q ∈ Q}$ is positive.  This guarantees that if the
input sequence is normal, then the output sequence is infinite and $𝒯(x)$
is well-defined.

Note that the output labels of the transitions in~$𝒯$ from
Theorem~\ref{thm:weighted} may have arbitrary lengths.  We first describe
the construction of an equivalent transducer~$𝒯'$ such that all output
labels in~$𝒯'$ have length at most~$1$. We call this transformation
\emph{normalization} and it consists in \emph{replacing} each transition
$p \trans{a | v} q$ in~$𝒯$ such that $|v| ⩾ 2$ by $n$ transitions:
\begin{displaymath}
   p \trans{a | b_1} q_1
     \trans{\emptyword | b_2} q_2
     ⋯
     q_{n-1} \trans{\emptyword | b_n} q
\end{displaymath}
where $q_1, q_2, …, q_{n-1}$ are new states and $v = b_1⋯ b_n$.
We refer to $p$ as the parent of $q_1, ⋯, q_{n-1}$.

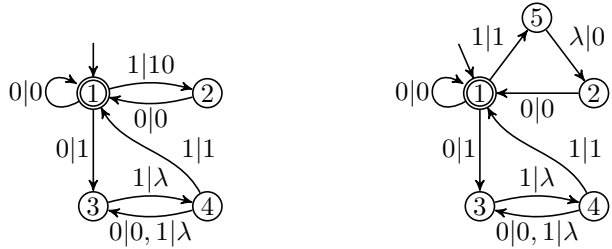
\begin{figure}[htbp]
  \begin{center}
  \begin{tikzpicture}[->,>=stealth',initial text=,semithick,auto,inner sep=1pt]
  \tikzstyle{every state}=[minimum size=0.4]
  \node[state,initial above,accepting] (q1) at (0,1.5) {$1$};
  \node[state] (q2) at (1.5,1.5) {$2$};
  \node[state] (q3) at (0,0) {$3$};
  \node[state] (q4) at (1.5,0) {$4$};
  \path (q1) edge[out=210,in=150,loop] node {$0|0$} (q1);
  \path (q1) edge [bend left=15] node {$1|10$} (q2);
  \path (q2) edge [bend left=15] node {$0|0$} (q1);
  \path (q1) edge node[swap] {$0|1$} (q3);
  \path (q3) edge [bend left=15] node {$1|\emptyword$} (q4);
  \path (q4) edge [bend left=15] node {$0|0,1|\emptyword$} (q3);
  \draw (q4) .. controls (1.1,0.75) and (0.4,0.75) .. node[swap,pos=0.25] {$1|1$} (q1);
  \end{tikzpicture}
  \hspace{2cm}
  \begin{tikzpicture}[->,>=stealth',initial text=,semithick,auto,inner sep=1pt]
  \tikzstyle{every state}=[minimum size=0.4]
  \node (init) at (-0.3,2.2) {};
  \node[state,accepting] (q1) at (0,1.5) {$1$};
  \node[state] (q2) at (1.5,1.5) {$2$};
  \node[state] (q3) at (0,0) {$3$};
  \node[state] (q4) at (1.5,0) {$4$};
  \node[state] (q5) at (0.75,2.5) {$5$};
  \path (init) edge (q1);
  \path (q1) edge[out=210,in=150,loop] node {$0|0$} (q1);
  \path (q1) edge node {$1|1$} (q5);
  \path (q2) edge node {$0|0$} (q1);
  \path (q5) edge node {$\emptyword|0$} (q2);
  \path (q1) edge node[swap] {$0|1$} (q3);
  \path (q3) edge[bend left=15] node {$1|\emptyword$} (q4);
  \path (q4) edge[bend left=15] node {$0|0,1|\emptyword$} (q3);
  \draw (q4) .. controls (1.1,0.75) and (0.4,0.75) .. node[swap,pos=0.25] {$1|1$} (q1);
  \end{tikzpicture}
  \end{center}
  \caption{The transducer $𝒯$ and its normalization $𝒯'$}
  \label{fig:normalized}
\end{figure}

To illustrate the construction, the normalized transducer obtained from the
transducer of Figure~\ref{fig:transducer} is pictured in
Figure~\ref{fig:normalized}.  State~$5$ has been added to split the
transition $1 \trans{1|10} 2$ into the finite run
$1 \trans{1|1} 5 \trans{\emptyword|0} 2$.

The main property of~$𝒯'$ is stated in the following lemma which
follows directly from the definition of normalization.
\begin{lemma} \label{lem:TequivTNorm}
  Both transducers $𝒯$ and~$𝒯'$ realize the same function, that is,
  $𝒯(x) = 𝒯'(x)$ for each sequence~$x$ in the domain of~$𝒯$.
\end{lemma}

From the normalized transducer~$𝒯'$ we construct a weighted automaton~$𝒜$
with the same state set as~$𝒯'$.  For all states $p,q$ and for every symbol
$b ∈ B$ the transition $p \trans{b} q$ is defined in~$𝒜$.  To assign
weights to transitions in $𝒜$, we auxiliary assign weights to transitions
in~$𝒯'$ as follows.  Each transition in~$𝒯'$ of the form $p \trans{a|v} q$
where $v$ is either a symbol or the empty word has weight $α_q/(\#A)α_p$.
Each transition in~$𝒯'$ of the form $p \trans{\emptyword|b} q$ (and
starting from a newly added state) has weight~$1$.  The sum of weights of
transitions starting from each state~$p$ is~$1$.  Indeed, if $p$ is a state
of~$𝒯$, the weights of transitions starting from~$p$ are the entries in
line indexed by~$p$ of the stochastic matrix~$P$.  If $p$ is a newly added
state, there is only one transition starting from~$p$ which has weight~$1$.
We now consider separately transitions that generate empty output from
those that do not.

Consider the $Q × Q$ matrix $E$ whose $(p,q)$-entry is given for each
pair $(p,q)$ of states by
\begin{displaymath}
  E_{p,q} = ∑_{a∈ A} \weight_{𝒯'}(p \trans{a|\emptyword} q).
\end{displaymath}
Let $E^*$ be the matrix defined by $E^* = ∑_{k ⩾ 0} E^k$, where $E^0$ is
the identity matrix.  By convention, there is indeed an empty run from~$p$
to~$p$ for each state~$p$ and this run has weight~$1$.  The entry
$E^*_{p,q}$ is the sum of weights of all finite runs with empty output
going from~$p$ to~$q$.  The matrix~$E^*$ can be computed because it is the
solution of the linear equation $E^* = EE^* + I$ where $I$ is the identity
matrix.  This proves in particular that all its entries are rational
numbers.

For each symbol $b ∈ B$ consider the $Q × Q$ matrix $D_b$ whose
$(p,q)$-entry is given for each pair $(p,q)$ of states by
\begin{displaymath}
  (D_b)_{p,q} = ∑_{a ∈ A ⊎ \{\emptyword\}}
  \weight_{𝒯'}(p \trans{a | b} q).
\end{displaymath}
We define the weight of a transition $p \trans{b} q$ in~$𝒜$ as
\begin{equation} \label{eq:weightsA}
  \weight_𝒜 (p \trans{b} q) = (E^*D_b)_{p,q}.  
\end{equation}
To assign initial weights to states we consider the matrix
$\hat{P} = ∑_{b ∈ B} {E^*D_b}$.  It is proved below in
Lemma~\ref{lem:matstoch} that this matrix is stochastic. The initial vector
of~$𝒜$ is the stationary distribution~$\hat{π}$ of~$\hat{P}$, that is, the
line vector~$\hat{π}$ such that $\hat{π}\hat{P} = \hat{π}$.  We assign to
each state~$q$ the initial weight $\hat{π}_q$.  Finally we assign final
weight~$1$ to all states.

We give below the matrices $E$, $E^*$, $D_0$, $D_1$ and~$\hat{P}$ and the
initial vector~$\hat{π}$ of the weighted automaton obtained from the
transducer pictured in Figure~\ref{fig:normalized}.

\begin{displaymath}
  E = 
  \left(
    \begin{array}{ccccc}
      0 & 0 & 0 & 0 & 0 \\
      0 & 0 & 0 & 0 & 0 \\
      0 & 0 & 0 & 1 & 0 \\
      0 & 0 & \frac{1}{4} & 0 & 0 \\
      0 & 0 & 0 & 0 & 0 \\
    \end{array}
  \right)
  \qquad
  E^* = 
  \left(
    \begin{array}{ccccc}
      1 & 0 & 0 & 0 & 0 \\
      0 & 1 & 0 & 0 & 0 \\
      0 & 0 & \frac{4}{3} & \frac{4}{3} & 0 \\
      0 & 0 & \frac{1}{3} & \frac{4}{3} & 0 \\
      0 & 0 & 0 & 0 & 1 \\
    \end{array}
  \right)
\end{displaymath}
\begin{displaymath}
  D_0 = 
  \left(
    \begin{array}{ccccc}
      \frac{1}{2} & 0 & 0 & 0 & 0 \\
      1 & 0 & 0 & 0 & 0 \\
      0 & 0 & 0 & 0 & 0 \\
      0 & 0 & \frac{1}{4} & 0 & 0 \\
      0 & 1 & 0 & 0 & 0 \\
    \end{array}
  \right)
  \qquad
  D_1 = 
  \left(
    \begin{array}{ccccc}
      0 & 0 & \frac{1}{4} & 0 & \frac{1}{4} \\
      0 & 0 & 0 & 0 & 0 \\
      0 & 0 & 0 & 0 & 0 \\
      \frac{1}{2} & 0 & 0 & 0 & 0 \\
      0 & 0 & 0 & 0 & 0 \\
    \end{array}
  \right)
\end{displaymath}
\begin{displaymath}
  \hat{P} = E^*(D_0 + D_1) = 
  \left(
    \begin{array}{ccccc}
      \frac{1}{2} & 0 & \frac{1}{4} & 0 & \frac{1}{4} \\
      1 & 0 & 0 & 0 & 0 \\
      \frac{2}{3} & 0 & \frac{1}{3} & 0 & 0 \\
      \frac{2}{3} & 0 & \frac{1}{3} & 0 & 0 \\
      0 & 1 & 0 & 0 & 0 \\
    \end{array}
  \right)
\end{displaymath}

\begin{proposition} \label{pro:correct}
  The automaton~$𝒜$ computes frequencies, that is, for every normal word
  $x$ and any finite word $w$ in~$B^*$, $\weight_{𝒜}(w) = \freq(𝒯(x),w)$.
\end{proposition}

The proof of the proposition requires some preliminary results.

Let us recall that a set of words~$L$ is called \emph{prefix-free} if no
word in~$L$ is a proper prefix of another word in~$L$.  As runs are defined
as sequences of (consecutive) transitions, this latter definition also
applies when $L$ is a set of runs.  We define $\freq(ρ, Γ)$ when $Γ$ is a
set of finite runs as follows.  Suppose that $ρ$ is the sequence
$τ_1τ_2τ_3⋯$ of transitions.  Then $\freq(ρ, Γ)$ is defined by
\begin{displaymath}
  \freq(ρ,Γ) = \lim_{n → ∞}
  \frac{\#\{ i < n : ∃ k ⩾ 0 \;\; τ_i ⋯ τ_{i+k} ∈ Γ\}}{n}.
\end{displaymath}
If $Γ$ is prefix-free (not to count twice the same start
position~$i$), the following equality holds
\begin{displaymath}
  \freq(ρ, Γ) = ∑_{γ∈Γ}{\freq(ρ,γ)}
\end{displaymath}
assuming that each limit of the right-hand sum does exist.  If $Γ$ is a set
of finite runs starting from the same state~$p$, the \emph{conditional
  frequency} of~$Γ$ in a run~$ρ$ is defined as the ratio between the
frequency of~$Γ$ in~$ρ$ and the frequency of~$p$ in~$ρ$, that is,
$\freq(ρ, Γ)/\freq(ρ, p)$.  Furthermore if $Γ$ is prefix-free, the
conditional frequency of~$Γ$ is the sum of the conditional frequencies of
its elements.

Let $x$ be a fixed normal word and let $ρ$ and~$ρ'$ be respectively the
runs in~$𝒯$ and~$𝒯'$ with label~$x$.  By Proposition~\ref{pro:statefreq},
the frequency $\freq(ρ, q)$ of each state~$q$ is $π_qα_q$ where $π$ and~$α$
and the left and right eigenvectors of the adjacency matrix of~$𝒯$ for the
eigenvalue~$1$.  The following lemma gives the frequency of states in~$ρ'$.

\begin{lemma} \label{lem:TTpfreq}
  There exists a constant~$C$ such that if $r$ is a state of~$𝒯$, then
  $\freq(ρ', r) = \freq(ρ, r)/C$ and if $r$ is newly created, then
  $\freq(ρ', r)$ is equal to $\freq(ρ', p)α_q/(\#A)α_p$ where $r$ comes
  from the splitting of a transition $p \trans{a|v} q$ in~$𝒯$.
\end{lemma}
\begin{proof}
  Observe that there is a one-to-one relation between runs labeled with
  normal words in~$𝒯$ and in~$𝒯'$.  More precisely, each transition $τ$
  in~$ρ$ is replaced by $\max(1, |v_τ|)$ transitions in~$ρ'$ (where $v_τ$
  is the output label of~$τ$).
  
  By Proposition~\ref{pro:runfreq}, each transition of~$𝒯$ has a frequency
  in~$ρ$.  The first result follows by taking
  $C = ∑_{τ}{\freq(ρ, τ) ⋅ \max(1, |v_τ|)}$ where the summation is taken
  over all transitions~$τ$ of~$𝒯$ and $v_τ$ is implicitly the output label
  of~$τ$.  The second result follows from Corollary~\ref{cor:condfreq}
  stating that each transition $p \trans{a|v} q$ has a conditional
  frequency of~$α_q/(\#A)α_p$ in~$ρ$.
\end{proof}

For each pair $(p,q)$ of states and each symbol $b ∈ B$, consider the set
$Γ_{p,b,q}$ of runs from~$p$ to~$q$ in~$𝒯'$ that have empty output labels
for all their transitions but the last one, which has $b$ as output label.
\begin{displaymath}
  Γ_{p,b,q} = \{ p \trans{a_1|\emptyword} ⋯
                   \trans{a_n|\emptyword} q_n \trans{a_{n+1}|b} q : n⩾ 0, q_i
                   ∈ Q, a_i ∈ A ∪ \{\emptyword\}\} 
\end{displaymath}
and let $Γ$ be the union $⋃_{p,q∈ Q, b ∈ B}{Γ_{p,b,q}}$.  Note that the
set~$Γ$ is prefix-free.  Therefore, the run~$ρ'$ has a unique factorization
$ρ = γ_0 γ_1 γ_2 ⋯$ where each $γ_i$ is a finite run in~$Γ$ and the ending
state of~$γ_i$ is the starting state of~$γ_{i+1}$.  Let $(p_i)_{i⩾0}$ and
$(b_i)_{i ⩾ 0}$ be respectively the sequence of states and the sequence of
symbols such that $γ_i$ belongs to $Γ_{p_i,b_i,p_{i+1}}$ for each $i ⩾ 0$.
Let us call $ρ''$ the sequence $p_0p_1p_2⋯$ of states of~$𝒯'$.

\begin{lemma}
  For each state~$q$ of~$𝒯'$, the frequency $\freq(ρ'', q)$ does exist.
\end{lemma}
\begin{proof}
  The sequence~$ρ''$ is a subsequence of the sequence of states in the
  run~$ρ'$.  An occurrence of a state~$q$ in~$ρ'$ is removed whenever the
  output of the previous transition is empty.
  
  Consider the transducer $\hat{𝒯}$ obtained by splitting each state $q$
  of~$𝒯$ into two states $q^\emptyword$ and $q^o$ in such a way that
  transitions with an empty output label end in a state~$q^\emptyword$ and
  other transitions end in a state~$q^o$.  Then each transition
  $p \trans{a|v} q$ is replaced by either the two transitions
  $p^\emptyword \trans{a|v} q^\emptyword$ and
  $p^o \trans{a|v} q^\emptyword$ if $v$ is empty or by the two transitions
  $p^\emptyword \trans{a|v} q^o$ and $p^o \trans{a|v} q^o$ otherwise.  The
  state~$q_0^\emptyword$ becomes the new initial state and non reachable
  states are removed.  Let $\hat{ρ}$ be the run in~$\hat{𝒯}$ labeled
  with~$x$.  By Proposition~\ref{pro:statefreq}, the frequencies
  $\freq(\hat{ρ}, q^\emptyword)$ and $\freq(\hat{ρ}, q^o)$ do exist.  Now
  consider the normalization $\hat{𝒯}'$ of~$\hat{𝒯}$ and the run $\hat{ρ}'$
  in~$\hat{𝒯}'$ labeled with~$x$.  It can be shown that the frequencies
  $\freq(\hat{ρ}', q^\emptyword)$ and $\freq(\hat{ρ}', q^o)$ do exist by an
  argument similar to the proof of Lemma~\ref{lem:TTpfreq}.  The sequence
  $ρ''$ is obtained from~$\hat{ρ}'$ by removing each occurrence of states
  $q^\emptyword$ and keeping occurrences of states~$q^o$.  It follows that
  the frequency of each state does exist in~$ρ''$.
\end{proof}

\begin{proof}[Proof of Proposition~\ref{pro:correct}]
  By Corollary~\ref{cor:condfreq}, the conditional frequency in~$ρ$ of each
  finite run~$γ$ of length~$n$ from~$p$ to~$q$ is $α_q/(\#A)^nα_p$.  It
  follows that the conditional frequency of each finite run~$γ'$ in~$ρ'$ is
  equal to its assigned weight in~$𝒯'$ while defining~$𝒜$.  By
  Equation~(\ref{eq:weightsA}), the weight of the transition
  $p \trans{b} q$ in~$𝒜$ is exactly the conditional frequency of the
  set~$Γ_{p,b,q}$ for each triple $(p,b,q)$ in $Q × B × Q$.  More
  generally, the product of the weights of the transitions
  $p_0 \trans{b_1} p_1 ⋯ p_{n-1} \trans{b_n} p_{n}$ is equal to the
  conditional frequency of the set $Γ_{p_0,b_1,p_1} ⋯ Γ_{p_n,b_n,p_{n+1}}$
  in~$ρ'$.
  
  It remains to prove that the frequency of each state~$q$ in~$ρ''$ is
  indeed its initial weight in the automaton~$𝒜$.  Let us recall that the
  initial vector of~$𝒜$ is the stationary distribution of the stochastic
  matrix~$\hat{P}$ whose $(p,q)$-entry is the sum
  $∑_{b∈ B}{\weight_{𝒜}(p \trans{b} q)}$, which is the conditional
  frequency of~$pq$ (as a word of length~$2$) in~$ρ''$.  It follows that
  the frequencies of states in~$ρ''$ must be the stationary of the
  matrix~$P$.
  
  Since the frequency of a word $v = b_1⋯ b_n$ in~$𝒯'(x)$ is
  the same as the sum over all sequences $p_0,p_1…,p_{n+1}$ of the
  frequencies of $Γ_{p_0,b_1,p_1} ⋯ Γ_{p_n,b_n,p_{n+1}}$
  in~$ρ'$, it is the weight of the word~$v$ in the
  automaton~$𝒜$.
\end{proof}

\begin{lemma} \label{lem:matstoch}
  The matrix $\hat{P} = ∑_{b ∈ B} {E^*D_b}$ which has been used to define
  the initial weights is stochastic.  Furthermore, its stationary
  distribution~$\hat{π}$ is proportional to the vector $π - π E$ where $π$
  is the stationary distribution of~$E + ∑_{b ∈ B} {D_b}$.
\end{lemma}
\begin{proof}
  It is an easy observation that if $P_1$ and~$P_2$ are two square matrices
  with non-negative coefficients such that $P_1 + P_2$ is stochastic, then
  $P_1^*P_2 = ∑_{n ⩾ 0}P_1^nP_2$ is stochastic.  Indeed, if $𝟏$ is the
  vector $(1,…,1)$, it is easily checked that
  \begin{displaymath}
    P_1^*P_2𝟏 = P_1^*(𝟏 - P_1𝟏) = 𝟏.
  \end{displaymath}
  It is also easy to check that
  \begin{displaymath}
    (π - π P_1) P_1^*P_2 = π P_2 = π - π P_1.
  \end{displaymath}
  In our case, the matrices
  $P_1$ and~$P_2$ come from the splitting of transitions of~$𝒯'$ into the
  ones with empty output and the ones with non-empty output.  They are
  respectively equal to $P_1 = E$ and $P_2 = ∑_{b ∈ B}D_b$.
\end{proof}

\begin{proof}[Proofs of Theorems \ref{thm:weighted} and \ref{thm:preservation}]
  To complete the proof of Theorems \ref{thm:weighted}
  and~\ref{thm:preservation}, we exhibit an algorithm deciding in cubic
  time whether an input deterministic transducer preserves normality.  Let
  $𝒯$ be an unambiguous transducer $⟨ Q,A,B,Δ,I,F ⟩$.  By definition, its
  size is the sum $∑_{τ ∈ Δ}|τ|$, where the size of a single transition
  $τ = p \trans{a|v} q$ is~$|τ| = |av|$.  We consider the alphabets to be
  fixed so they are not taken into account when computing complexity.

  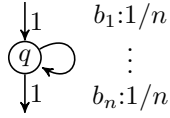
\begin{figure}[htbp]
    \begin{center}
      \begin{tikzpicture}[->,>=stealth',semithick,auto,inner sep=1.5pt]
      \tikzstyle{every state}=[minimum size=0.4]
      \node[state] (q6) at (0,0) {$q$} ;
      \node (qinit) at (0,0.8) {} ;
      \node (qfin) at (0,-0.8) {} ;
      \path (qinit) edge node {$1$} (q6);
      \path (q6) edge node {$1$} (qfin);
      \path (q6) edge[out=30,in=330,loop] node
        {$\begin{array}{c} b_1{:}1/n \\ \vdots \\ b_n{:}1/n \end{array}$} (q6) ;
      \end{tikzpicture}
    \end{center}
    \caption{Weighted automaton $ℬ$ such that $\weight_{ℬ}(w) = 1/(\#A)^{|w|}$}
    \label{fig:weighted3}
  \end{figure}
  
  By Proposition~\ref{pro:decomposition}, the algorithm decomposes the
  transducer into strongly connected components and make a list of all
  strongly connected components with a final state and spectral radius~$1$.
  This latter condition is checked in cubic time by computing the
  determinant of its adjacency matrix minus the identity matrix.  For each
  strongly connected component in the list, the algorithm checks whether
  it preserves normality or not.  This is achieved by computing the
  weighted automaton~$𝒜$ and checking that the weight of each word~$w$ is
  $1/(\#A)^{|w|}$.  This latter step is performed by comparing $𝒜$ with the
  weighted automaton~$ℬ$ such that $\weight_{ℬ}(w) = 1/(\#A)^{|w|}$.  The
  automaton~$ℬ$ is pictured in Figure~\ref{fig:weighted3}.
\begin{itemize}
\item[]
\textbf{Input:} $𝒯 = ⟨ Q,A,B,Δ,I,F ⟩$ 
  an input deterministic complete transducer.\\
\textbf{Output:} \texttt{True} if $𝒯$ preserves normality and
  \texttt{False} otherwise.\\
\textbf{Procedure:}
\begin{enumerate}
  \item[I.] Compute the strongly connected components of $𝒯$
  \item[II.] For each strongly connected component $S_i$ $𝒯$:
    \begin{enumerate}
    \item[1.] Compute the normalized transducer $𝒯'$, equivalent to $S_i$.
    \item[2.] Use $𝒯'$ to build the weighted automaton $𝒜$:
      \begin{enumerate}
      \item[a.] Compute the weights of the transitions of~$𝒜$.
        \begin{itemize}
        \item[] Compute the matrix $E$                          
        \item[] Compute the matrix $E^*$ solving $(I-E)X=I$  
        \item[] For each $b ∈ B$, for each $p,q ∈ Q$:
          \begin{itemize}
          \item[] compute the matrix $D_b$
          \item[] define the transition $p \trans{b} q$ with
            weight $(E^*D_b)_{p,q}$.
          \end{itemize} 
        \end{itemize} 
      \item[b.] Compute the stationary distribution $π$ of
        the Markov chain induced by~$𝒜$.
      \item[c.] Assign initial weight $π[i]$ to each state $i$,
        and let final weight be $1$ for all states.
      \end{enumerate}
    \item[3.] Compare $𝒜$ against the automaton~$ℬ$ using Schützenberger's
      algorithm \cite{CardonCrochemore80,Sakarovitch09b} to check whether
      they realize the same function.
    \item[4.] If they do not compute the same function, return \texttt{False}.
    \end{enumerate}
  \item[III.] Return \texttt{True}
\end{enumerate}
\end{itemize}

Now we analyze the complexity of the algorithm.  Computing recurrent
strongly connected components can be done in time $O(\#Q^2) ⩽ O(n^2)$
using Kosaraju's algorithm if the transducer is implemented with an
adjacency matrix \cite[Section~22.5]{Cormen09}.

We refer to the size of the component~$S_i$ as~$n_i$.  The cost of
normalizing the component is $O(n_i^2)$, mainly from filling the new
adjacency matrix.  The most expensive step when computing the transitions
and their weight is to compute $E^*$.  The cost is $O(n_i^3)$ to solve
the system of linear equations.  To compute the weights of the states we
have $O(n_i^3)$ to solve the system of equations to find the stationary
distribution.  Comparing the automaton to the one computing the expected
frequencies can be done in time $O(n_i^3)$ \cite{CardonCrochemore80} since
the coefficients of both automata are in $ℚ$.
\end{proof}

\section{Preservation of normality by selection}
\label{sec:selection}

The aim of this section is to show that in the case of selectors, the
weighted automaton given by the construction detailed in the previous
section has a special form.  This allows us to give another evidence that
oblivious prefix selection preserves normality.

\subsection{Oblivious prefix selection}

A \emph{selector} is a deterministic transducer such that each of its
transitions has one of the types $p \trans{a|a} q$ (type~I),
$p \trans{a|\emptyword} q$ (type~II) for a symbol $a ∈ A$. In a selector,
the output of a transition is either the symbol read by the transition
(type~I) or the empty word (type~II).  Therefore, it can be always assumed
that the output alphabet~$B$ is the same as the input alphabet~$A$.  It
follows that for each run $p \trans{u|v} q$, the output label~$v$ is a
subword, that is a subsequence, of the input label~$u$.  A selector is
pictured in Figure~\ref{fig:transducer2}.

Let us recall the link between oblivious prefix selection and selectors.
Let $x = a_1a_2a_3 ⋯ $ be a sequence over the alphabet~$A$.  Let $L ⊆ A^*$
be a set of finite words over~$A$. The word obtained by \emph{oblivious
  prefix selection} of~$x$ by~$L$ is
$x \prefsel L = a_{i_1}a_{i_2}a_{i_3} ⋯$ where $i_1,i_2,i_3,…$ is the
enumeration in increasing order of all the integers~$i$ such that the
prefix $a_1 a_2 ⋯ a_{i-1}$ belongs to~$L$.  This selection rule is called
\emph{oblivious} because the symbol~$a_i$ is not included in the considered
prefix.  If $L = A^*1$ is the set of words ending with a~$1$, the
sequence~$x \prefsel L$ is made of all symbols of~$x$ occurring after a~$1$
in the same order as they occur in~$x$.  If $L ⊆ A^*$ is a rational set,
the oblivious prefix selection by~$L$ can be performed by an oblivious
selector.  There is indeed an oblivious selector~$𝒮$ such that for each
input word~$x$, the output $𝒮(x)$ is the result $x \prefsel L$ of the
selection by~$L$.  This selector~$𝒮$ can be obtained from any deterministic
automaton~$𝒜$ accepting~$L$.  Replacing each transition $p \trans{a} q$
of~$𝒜$ by either $p \trans{a|a} q$ if the state~$p$ is accepting or by
$p \trans{a|\emptyword} q$ otherwise yields the selector~$𝒮$.  It can be
easily verified that the obtained transducer is an oblivious selector
performing the oblivious prefix selection by~$L$.

\begin{lemma}
  For each symbol~$a$, the matrix~$E^*D_a$ of the construction satisfies
  \begin{displaymath}
    E^*D_a 𝟏 = \frac{1}{\#A}𝟏
  \end{displaymath}
  where $𝟏$ is the vector $(1, …, 1)$.
\end{lemma}
\begin{proof}
  In a deterministic automaton, each sequence is the label of an infinite
  run starting from each state~$q$.  This means that $α_q$ is equal to~$1$
  for each state~$q$.  We first consider the matrix~$D_a$ for each
  symbol~$a$ in the alphabet $A = B$.  Let $p$ be fixed state of the
  selector.  If the transitions starting from~$p$ have type~I, them the
  entry $(D_a)_{p,q}$ is equal to $\#\{ a : p \trans{a} q\}/\#A$ for each
  state~$q$.  If the transitions starting from~$p$ have type~I, them the
  entry $(D_a)_{p,q}$ is equal to zero for each state~$q$.  In the former
  case $∑_{q ∈ Q} (D_a)_{p,q} = 1$ and in the latter case
  $∑_{q ∈ Q} (D_a)_{p,q} = 0$.  Note that this sums do not depend on the
  symbol~$b$.  It follows that for symbols $a$ and~$b$
  \begin{displaymath}
    ∑_{q ∈ Q} (E^*D_a)_{p,q} = ∑_{r ∈ Q}E^*_{p.r}∑_{q ∈ Q} (D_a)_{r,q}
   =  ∑_{r ∈ Q}E^*_{p.r}∑_{q ∈ Q} (D_b)_{r,q}  = ∑_{q ∈ Q} (E^*D_b)_{p,q}.
  \end{displaymath}
  Since the matrix $\hat{P} = ∑_{b ∈ B} {E^*D_b}$ is stochastic by
  lemma~\ref{lem:matstoch}, the sum $∑_{q ∈ Q} (E^*D_a)_{p,q}$ is equal to
  $1/\#A$ for each symbol $a ∈ A$ and each state $p ∈ Q$. This is exactly
  the claimed equality.
\end{proof}

As a corollary, we get Agafonov's theorem \cite{Agafonov68}.
\begin{corollary}
  Oblivious prefix selection by a rational set preserves normality.
\end{corollary}
\begin{proof}
  The weight $\weight_{𝒜}(w)$ computed by the weighted automaton~$𝒜$ for a
  word $w = a_1 ⋯ a_n$ is equal to
  \begin{displaymath}
    \hat{π} E^*D_{a_1} E^* D_{a_2} ⋯ E^* D_{a_n} 𝟏
  \end{displaymath}
  where $\hat{π}$ is the stationary distribution of~$\hat{P}$.  By the
  previous lemma, this is is equal to $1/(\#A)^n$.
\end{proof}

\subsection{Non-oblivious prefix selection}

Let $x = a_1a_2a_3 ⋯ $ be an infinite word over alphabet~$A$.  Let
$L ⊆ A^*$ be a set of finite words over~$A$. The word obtained by
\emph{non-oblivious prefix selection} of $x$ by $L$ is
$x \noprefsel L = a_{i_1}a_{i_2}a_{i_3} ⋯$ where $i_1,i_2,i_3,…$ is the
enumeration in increasing order of all the integers~$i$ such that the
prefix $a_1 a_2a_3 ⋯ a_i$ including $a_i$ belongs to~$L$.  Non-oblivious
selection does not preserve normality in general.  If $L = A^*1$ is the
set of words ending with a~$1$, the non-oblivious selection selects only
$1$s.

Each deterministic automaton accepting the set~$L$ can be turn into a
selector performing the selection by~$L$.  Replacing each transition
$p \trans{a} q$ by a transition $p \trans{a|a} q$ if the state~$q$ is final
and by $p \trans{a|\emptyword} q$ otherwise and keeping everything else
unchanged yields a selector.  Note that this selector might not be
oblivious.

We now introduce a classical class of rational sets called group sets.  A
\emph{group automaton} is a deterministic automaton such that each symbol
induces a permutation of the states.  By inducing a permutation, we mean
that, for each symbol~$a$, the function which maps each state~$p$ to the
state~$q$ such that $p \trans{a} q$ is a permutation of the state set.  Put
another way, if $p \trans{a} q$ and $p' \trans{a} q$ are two transitions of
the automaton, then $p = p'$.  A rational set~$L ⊆ A^*$ is called a
\emph{group set} if $L$ is recognized by a group automaton.  It is well
know that a rational set is a group set if and only if its syntactic
monoid is a group.

The following lemma is straightforward.
\begin{lemma} \label{lem:uniformgroup}
  The stationary distribution of a group automaton is the uniform
  distribution.
\end{lemma}

\begin{lemma}
  For each symbol~$a$, the matrix~$E^*D_a$ of the construction satisfies
  \begin{displaymath}
    𝟏D_a  = \frac{1}{\#A}(𝟏 - 𝟏E)
  \end{displaymath}
  where $𝟏$ is the vector $(1, …, 1)$.
\end{lemma}
\begin{proof}
  We consider the matrix~$D_a$ for each symbol~$a$ in the alphabet $A = B$.
  Let $q$ be fixed state of the selector.  If the transitions ending in~$q$
  have type~I, them the entry $(D_a)_{p,q}$ is equal to
  $\#\{ a : p \trans{a} q\}/\#A$ for each state~$p$.  If the transitions
  ending in~$q$ have type~I, them the entry $(D_a)_{p,q}$ is equal to zero
  for each state~$q$.  In the former case $∑_{p ∈ Q} (D_a)_{p,q} = 1$ and in
  the latter case $∑_{p ∈ Q} (D_a)_{p,q} = 0$.  Note that this sums do not
  depend on the symbol~$b$.  $D_a = ∑_{b ∈ A}{D_b}/\#A$.  The claimed
  equality follows from the fact that $𝟏$ is proportional to the stationary
  distribution of $E + ∑_{b ∈ A}{D_b}$ by Lemma~\ref{lem:uniformgroup}.
\end{proof}

The following result states that if $L$ is a group set, the non-oblivious
selection by~$L$ preserves normality.  It has been proved in
\cite{CartonVandehey19}.

\begin{corollary} 
  Non-oblivious prefix selection by a rational group set preserves normality.
\end{corollary}
\begin{proof}
  By Lemmas \ref{lem:uniformgroup} and~\ref{lem:matstoch}, the stationary
  distribution of the weighted automaton~$\mathcal{A}$ is $C(𝟏 - 𝟏E)$ where
  the constant~$C$ is chosen such that the coordinates of $C(𝟏 - 𝟏E)$ sum
  up to~$1$.  The weight $\weight_{𝒜}(w)$ computed by the weighted
  automaton~$𝒜$ for a word $w = a_1 ⋯ a_n$ is equal to
  \begin{displaymath}
    C(𝟏 - 𝟏E) E^*D_{a_1} E^* D_{a_2} ⋯ E^* D_{a_n} 𝟏 
  \end{displaymath}
  which is equal to to $1/(\#A)^n$ by the previous lemma.
\end{proof}

\section*{Conclusion}

The first result of the paper provides a weighted automaton which gives the
limiting frequency of each block in the output of a normal input.  This
automaton can be used to check another property of this invariant.  It can
be decided, for instance, whether this measure is a Bernoulli measure.
This boils down to checking whether the minimal automaton has only a single
state.

In this work, it is assumed that the input of the transducer is normal,
that is generic for the uniform measure.  It seems that the results can be
extended to the more general setting of Markovian measures.  The case of
hidden Markovian measure, that is, measures computed by weighted automata,
seems however more involved \cite{HanselPerrin90}.

\section*{Acknowledgements}

The author would like to thank Verónica Becher for many fruitful
discussions and suggestions.  The author is a member of the IRP SINFIN,
CONICET/Universidad de Buenos Aires--CNRS/Université de Paris and he is
supported by the ECOS project PA17C04.  The author is also partially funded
by the DeLTA project (ANR-16-CE40-0007).

\bibliographystyle{plain}
\bibliography{algorithm}

\end{document}